\documentclass[onecolumn,draftclsnofoot]{IEEEtran}

\usepackage{subfigure}
\usepackage{epsfig,graphicx,psfrag,dsfont}
\usepackage{amsfonts,amsmath,amssymb,color}
\usepackage{bbm,setspace,amsthm,cite}

\usepackage{xspace}
\usepackage{bbm}

\newcommand{\ex}{{\rm e}}
\newcommand{\Ac}{\mathcal{A}}
\newcommand{\Bc}{\mathcal{B}}
\newcommand{\Cc}{\mathcal{C}}

\newcommand{\Ec}{\mathcal{E}}
\newcommand{\Fc}{\mathcal{F}}

\newcommand{\Nc}{\mathcal{N}}

\newcommand{\Xc}{\mathcal{X}}
\newcommand{\Yc}{\mathcal{Y}}
\newcommand{\Zc}{\mathcal{Z}}

\newcommand{\Xh}{{\hat{X}}}
\newcommand{\Yh}{{\hat{Y}}}

\newcommand{\qh}{{\hat{q}}}

\newcommand{\xh}{{\hat{x}}}

\newcommand{\xt}{{\tilde{x}}}

\def\a{\alpha}
\def\b{\beta}
\def\g{\gamma}
\def\d{\delta}
\def\e{\epsilon}

\def\l{\lambda}

\DeclareMathOperator\E{E}
\let\P\relax
\DeclareMathOperator\P{P}

\newcommand\ie{i.e.,\xspace}
\def\textiid{i.i.d.\@\xspace}
\newcommand\iid{\ifmmode\text{ i.i.d. } \else \textiid \fi}

\newcommand{\ind}{\mathbbmss{1}}

\DeclareMathOperator*{\argmin}{arg\,min}

\title{Universal Compressed Sensing}

\author{Shirin Jalali, H. Vincent Poor}

\date{} 

\begin{document}

\newtheorem{property}{Property}
\newtheorem{question}{Question}
\newtheorem{claim}{Claim}
\newtheorem{guess}{Conjecture}
\newtheorem{definition}{Definition}
\newtheorem{fact}{Fact}
\newtheorem{assumption}{Assumption}
\newtheorem{theorem}{Theorem}
\newtheorem{lemma}{Lemma}
\newtheorem{remark}{Remark}
\newtheorem{ctheorem}{Corrected Theorem}
\newtheorem{corollary}{Corollary}
\newtheorem{proposition}{Proposition}
\newtheorem{example}{Example}

\maketitle
\begin{abstract}
The main promise of compressed sensing is  accurate recovery of high-dimensional structured signals from an underdetermined set of randomized linear projections. Several types of structure such as sparsity and low-rankness have already been explored in the literature. For each type of structure, a recovery algorithm has been developed  based on the properties of  signals with that type of structure. Such  algorithms recover any signal   complying with the assumed model from its sufficient number of linear projections. However, in many practical situations the underlying structure of the signal is not known, or is only partially known. Moreover, it is desirable to have recovery algorithms that can be applied to signals with different types of structure.

In this paper, the problem of developing  universal algorithms for compressed sensing of stochastic processes is studied.  First, R\'enyi's notion of information dimension (ID) is generalized to analog stationary  processes. This provides a measure of complexity for such  processes and is connected  to  the number of measurements required for their   accurate recovery. Then a minimum entropy pursuit (MEP) optimization approach  is proposed, and it is proven that it can reliably recover any  stationary process satisfying some mixing constraints  from sufficient number of randomized linear measurements, without having any prior information about the  distribution of the process. It is proved that a Lagrangian-type approximation of the MEP optimization problem, referred to as Lagrangian-MEP problem, is identical to a  heuristic implementable algorithm proposed by Baron et al. It is shown that for the right choice of parameters the Lagrangian-MEP algorithm, in addition to having  the same asymptotic  performance as  MEP optimization,   is also robust to the measurement noise. For memoryless sources with a discrete-continuous mixture distribution, the  fundamental limits of the minimum number of required measurements by a non-universal compressed sensing decoder is characterized by Wu et al. For such sources, it is proved that there is no loss in universal coding, and both the MEP and the Lagrangian-MEP  asymptotically achieve the optimal performance.
\end{abstract}

\section{Introduction}

Consider the fundamental problem of compressed sensing (CS): a signal $x_o^n\in\mathds{R}^n$ is measured through a data acquisition process  modeled by  a linear projection system: $y_o^m=Ax_o^n$, where $A\in\mathds{R}^{m\times n}$ denotes the measurement matrix. The signal $x_o^n$ is usually  high-dimensional, and  the number of measurements is much smaller than the ambient dimension of the signal, \ie $m\ll n$. The decoder is interested in recovering $x_o^n$ from the  measurements $y_o^m$.  Since the system of linear equations described by $y_o^m=Ax^n$ has infinitely many solutions, without any side information, clearly it is impossible to recover $x_o^n$ from $y_o^m$. However, with some extra  information about the structure of $x_o^n$,  one might be able to   recover  $x_o^n$ from $y_o^m$ reliably. Intuitively, having this extra information enables the decoder to, among the signals that satisfy the measurement constraints,  search for the one that is (more) consistent with the structure model. For instance for sparse signals, \ie signal with $k=\|x_o^n\|_0\ll n$,\footnote{For $x^n\in\mathds{R}^n$, $\|x_o^n\|_0\triangleq |\{i: x_i\neq 0\}|$.} the decoder might try to find the signal with minimum $\ell_0$-norm among the signals that satisfy $Ax^n=y_o^m$. In that case the reconstruction signal will be
\[
\xh^n_o=\argmin_{Ax^n=y_o^n}\|x^n\|_0.
\]
While this optimization is not  practical, it has well-known approximations that can  be implemented efficiently \cite{Donoho:06,CandesT:06,CandesR:06}.  This result can also be extended to several other structures such as group-sparsity and low-rankness. (Refer to \cite{RichModelbasedCS, ChRePaWi10, VeMaBl02, ReFaPa10, ShCh11, HeBa12, HeBa11, DoKaMe06} for some examples of the other structures studied in the literature.) In each of these cases, the signal is known to have some specific structure, and the decoder exploits this side information  to recover the signal from its under-sampled set of linear projections.

Structures that are already studied in the compressed sensing literature are often simple models such as sparsity. However, natural signals typically exhibit  much more complicated and diverse patterns. Therefore, it is desirable to have a recovery algorithm that can be applied to sources with diverse structures without having some prior information about the source model. Such algorithms are referred to as universal algorithms   in the information theory literature. More formally universal algorithms are defined as algorithms that achieve the optimal performance without knowing the source distribution. Existence of such algorithms has been proved for several different problems such as compression  \cite{LZ77, Sakrison:70,Ziv:72,NeuhoffG:75,NeuhoffS:78}, denoising \cite{kolmogrov_sampler,dude} and prediction \cite{MerhavGutmanFeder92,MerhavF:98}.

In order to  develop a universal compressed sensing algorithm,  there are some fundamental questions that need to be addressed: What does it mean for an analog\footnote{Throughout the paper,  an analog signal refers to a continuous-alphabet discrete-time signal.} signal to be of low complexity or structured?  How can the structure or the complexity of an analog signal be measured?  Is it possible to design a universal compressed sensing decoder that is able to recover structured signals from their randomized linear projections\footnote{Throughout the paper, linear measurements acquired by a measurement matrix generated from a random distribution is denoted by randomized linear projections or randomized linear measurements.} without knowing the underlying structure of the signal?

The problem of  universal  compressed sensing has already been studied in  the literature  \cite{BaronD:11,BaronD:12,JalaliM:11,JalaliM:12,JalaliM:14}.  In \cite{BaronD:11} and \cite{BaronD:12}, the authors propose a heuristic implementable  algorithm for universal compressed sensing of stochastic processes. In \cite{JalaliM:11}, the authors define the Kolmogorov information dimension (KID) of a  deterministic analog  signal as a  measure of its complexity.  The KID of a signal $x_o^n$ is defined   as the growth rate of the Kolmogorov complexity of the quantized version of $x_o^n$ normalized by the $\log$ of the number of quantization levels, as the number of quantization levels grows to infinity. Employing this measure of complexity, the authors in \cite{JalaliM:11} and \cite{JalaliM:14}   propose a minimum complexity pursuit (MCP) optimistion as a universal signal recovery decoder.  MCP is based on Occam's razor \cite{occam}, \ie among all  signals  satisfying the linear measurement constraints, MCP seeks the one with the lowest complexity.  While MCP proves the existence of universal compressed sensing algorithms, it is not an implementable algorithm, since it is based on minimizing Kolmogorov complexity \cite{Solomonoff,KolmogorovC}, which is not computable.

In this paper we focus on stochastic signals  and develop an implementable algorithm for universal compressed sensing of stochastic processes. To achieve this goal,  we first need to develop a measure of complexity for stochastic  processes that differentiates between different processes in terms of their complexities. To define such a measure,  we   extend the R\'enyi's notion of the information dimension of an analog random variable \cite{Renyi:59}  to define the information dimension of a stochastic process. As we will show, this extension is consistent with   R\'enyi's information dimension such that the   information dimension of a  memoryless stationary  process $X=\{X_i\}_{i=1}^{\infty}$ is equal to the R\'enyi information dimension of its first-order marginal distribution ($X_1$). It has recently been  proved  that for independent and identically distributed (i.i.d.) processes with a mixture of discrete-continuous distribution, their R\'enyi information dimension  characterizes the fundamental limits of (non-universal) compressed sensing \cite{WuV:10}.

Again consider the basic problem of compressed sensing:  $X_o^n$ is generated by an analog stationary process $X=\{X_i\}_{i=1}^{\infty}$,  the decoder observes its linear projections $Y_o^m=AX_o^n$, where $m< n$, and is interested in recovering $X_o^n$. To recover $X_o^n$ from $Y_o^m$, in the same spirit of the MCP algorithm, we propose  minimum entropy pursuit (MEP) optimization, which among all the signals $x^n$ satisfying the measurement constraint $Y_o^n=Ax^n$, outputs the one whose quantized version has the minimum conditional empirical entropy. We prove that, asymptotically, for a proper choice of the quantization level and the order of the conditional empirical entropy,  and having slightly more than  the (upper) information dimension of the process times the ambient dimension of the process randomized linear measurements, MEP presents an asymptotically lossless estimate of $X_o^n$. While MEP is not easy to implement, we also present an implementable version with the same asymptotic performance guarantees as MEP. The implementable approximation of the MEP optimization, which we refer to as Lagrangian-MEP,  is identical  to the heuristic algorithm proposed and implemented in  \cite{BaronD:11} and \cite{BaronD:12} for universal compressed sensing. We  prove that for the right choice of parameters,  the Lagrangian-MEP algorithm has the same asymptotic performance as   MEP and in addition is also robust to  measurement noise. That is, the asymptotic performance of the Lagrangian-MEP algorithm does not change when the measurement vector is corrupted by a small-enough measurement noise vector. For memoryless sources with a discrete-continuous mixture distribution, we  show that there is no loss in the performance due to universal coding, and both the MEP optimization and the Lagrangian-MEP algorithm achieve the optimal performance derived in \cite{WuV:10}.

The organization of the paper is as follows. Section \ref{sec:background} introduces the notation used in the paper and reviews some related background. Section \ref{sec:overview} presents an overview of the  main contributions of the paper.  In Section \ref{sec:ID-ergodic},  we first generalize   the  R\'enyi's notion of the information dimension of a random variable \cite{Renyi:59}  and define the information dimension of a stationary process. In Section \ref{sec:universal-cs}, we introduce the MEP  optimization for universal compressed sensing, and  also provide an implementable version of MEP, namely Lagrangian-MEP, which is the same heuristic algorithm proposed in \cite{BaronD:11} and \cite{BaronD:12},  and prove its optimality and its robustness to measurement noise.  The proofs of all of the main results are given in Section \ref{sec:proof}.  Section \ref{sec:conclusion} concludes the paper.

\section{Background}\label{sec:background}

In this section we first introduce the notation that is used throughout the paper. Then, we review two basic concepts, namely  empirical distribution and universal lossless compression, which we employ in developing our results on universal compressed sensing.

\subsection{Notation}
Calligraphic letters such as $\Xc$ and $\Yc$ denote sets. For a finite set $\Xc$, let $|\Xc|$ denote the size of  $\Xc$.  Given vectors $u^n,v^n\in\mathds{R}^n$, let $\langle u^n,v^n \rangle$ denote their inner product, \ie $\langle u^n,v^n \rangle\triangleq \sum_{i=1}^nu_iv_i$. Also, $\| u^n\|_2\triangleq (\sum_{i=1}^nu_i^2)^{0.5}$ denotes   the $\ell_2$-norm of $u^n$. For $1\leq i\leq j \leq n$, $u_i^j \triangleq (u_i,u_{i+1},\ldots,u_j)$. To simplify the notation, $u^j\triangleq u_1^j$.   The set of all finite-length binary sequences is denoted by  $\{0,1\}^*$, \ie $\{0,1\}^*\triangleq\cup_{n\geq 1}\{0,1\}^n$. Similarly, $\{0,1\}^{\infty}$ denotes the set of infinite-length binary sequences. Throughout the paper $\log$ refers to logarithm to the basis of $2$ and $\ln$ refers to the natural logarithm.

Random variables are represented  by upper-case letters such as $X$ and $Y$. The alphabet of the random variable $X$ is denoted by $\Xc$. Given a sample space $\Omega$ and  event  $\Ac\subseteq \Omega$, $\ind_{\Ac}$ denotes the indicator function of   $\Ac$. Given $x\in\mathds{R}$, $\delta_x$ denotes the Dirac measure with an atom at $x$.

Given a real number $x\in\mathds{R}$,  $\lfloor x\rfloor$  ($\lceil x\rceil$) denotes the largest (the smallest) integer number smaller (larger) than $x$. Further,  $[x]_b$ denotes the $b$-bit quantized version of $x$ that results  from taking the first $b$ bits in the binary expansion of $x$. That is, for   $x=\lfloor x\rfloor +\sum_{i=1}^{\infty}2^{-i}(x)_i$, where $(x)_i\in\{0,1\}$,
\[
[x]_b\triangleq \lfloor x\rfloor+\sum_{i=1}^{b}2^{-i}(x)_i.
\]
Also, for $x^n \in\mathds{R}^n$, define
\[
[x^n]_b\triangleq ([x_1]_b,\ldots,[x_n]_b).
\]
For a positive integer $\ell$, let 
\[
\langle x \rangle_{\ell}\triangleq { \lfloor {\ell }x \rfloor \over \ell}.
\]
By this definition, $\langle x \rangle_{\ell}$ is a finite-alphabet approximation of the random variable $X$, such that $0< x- \langle x \rangle_{\ell} \leq {1\over \ell}.$

%
\subsection{Conditional empirical entropy}\label{sec:cond-ent}

 Consider a stochastic process $X=\{X_i\}_{i=1}^{\infty}$, with finite alphabet  $\Xc$ and probability measure $\mu(\cdot)$. The entropy rate of  a  stationary  process $X$ is defined as
\begin{align}
\bar{H}(X)\triangleq \lim_{n\to \infty}{H(X_1,\ldots,X_n)\over n}.\label{eq:entropy-rate}
\end{align}

The $k$-th order empirical distribution induced by $x^n\in\Xc^n$, $p_k(.|x^n)$ is defined as
\[
p_k(a^k|x^n)={|\{i: x_{i-k}^{i-1}=a^k, 1\leq i\leq n\}| \over n},
\]
where we make a circular  assumption such that $x_{j}=x_{j+n}$, for $j\leq 0$.

\begin{definition}
The conditional empirical entropy induced by $x^n\in\Xc^n$, $\hat{H}_k(x^n)$, is equal to  $H(U_{k+1}|U^k)$, where $U^{k+1}\sim p_{k+1}(\cdot|x^n)$.
 \end{definition}

For a stationary finite-alphabet process $X$,  if $k$ grows to infinity as $k=o(\log n)$,  $\hat{H}_k(X^n)$ converges, almost surely, to the entropy rate of the process $X$, \ie $\hat{H}_k(X^n)\to \bar{H}(X)$, almost surely \cite{Shields:96}. Therefore, if we fix the size of the source alphabet,   $\hat{H}_k$ is a universal estimator of the source entropy rate, which in turn is a measure   of the source complexity.


\subsection{Universal lossless compression}

One of the most well-studied universal coding problems is the problem of universal lossless compression. The entropy rate of a stationary process characterizes the minimum number of bits per symbol required for its lossless compression. Universal lossless compression algorithms, such as the Lempel-Ziv  algorithm \cite{LZ}, asymptotically spend the same number of bits per symbol for compressing any stationary ergodic process, without knowing its distribution. To design a universal lossless compression algorithm, similar to universal compressed sensing, one needs to develop a universal measure of complexity. The fundamental difference between the two problems is that while compressed sensing is mainly concerned with continuous-alphabet sources, lossless compression is concerned  with  discrete-alphabet processes. In the following, we briefly review the mathematical definition of a universal lossless compression  algorithm.

Consider the problem of universal lossless compression of discrete stationary ergodic sources described as follows. A family of source codes $\{\Cc_n\}_{n\geq 1}$ consists of a sequence of codes corresponding to different blocklengths. Each code $\Cc_n$ in this family is defined by an encoder  function $f_n$ and a decoder function $g_n$ such that
\[
f_n:\Xc^n\to \{0,1\}^{*},
\]
and
\[
g_n:\{0,1\}^{*}\to\hat{\Xc}^n.
\]
Here $\hat{\Xc}$ denotes the reconstruction alphabet which is also  assumed to be discrete and in many cases is equal to $\Xc$.  The encoder $f_n$ maps each source block $X^n$ to a binary sequence of finite length, and the decoder $g_n$ maps the coded bits back to the signal space as $\hat{X}^n=g_n(f_n(X^n))$. Let  $l_n(f_n(X^n))=|f_n(X^n)|$ denote the length of the binary sequence assigned to the sequence $X^n$. We assume that the codes are lossless (non-singular), \ie $f_n(x^n)\neq f_n(\xt^n)$, for all $x^n\neq \xt^n$. A family of lossless codes is called universal, if
\[
 {1\over n}\E[l_n(f_n(X^n))]\to \bar{H}(X),
 \]
and $\P(X^n\neq \Xh^n ) \to 0$, as $n$ grows to infinity,  for any discrete  stationary  process $X$.  A family of lossless codes is called \emph{point-wise} universal, if
\[
 {1\over n}l_n(f_n(X^n))\to \bar{H}(X),
 \]
almost surely,  for any discrete  stationary ergodic process $X$. The   Lempel-Ziv  algorithm \cite{LZ} is both a universal and a point-wise universal lossless compression algorithm. For a discrete-alphabet sequence $u^n$, $\ell_{\rm LZ}(u^n)$ denotes the length of encoded version of $u^n$ by the Lempel-Ziv  algorithm.


\section{Overview of results}\label{sec:overview}

In recovering a structured signal $x_o^n$ from undersampled linear measurements $y_o^m=Ax_o^n$, $m<n$, non-universal algorithms look for the signal that complies  with both the measurements  and the known structure.  For several types of structure, it has been proved that with enough  measurements, this procedure yields a reliable estimate of the input vector $x_o^n$. In the universal compressed sensing problem, the decoder does not have any information about the structure or the distribution of the source. In the deterministic settings, the authors in \cite{JalaliM:14}  proved  that universal compressed sensing is possible and proposed the KID measure of complexity and the MCP universal recovery algorithm. However, as explained earlier,  both  the MCP optimization and the KID measure are based on the notion of Kolmogorov complexity and hence are not computable. Moreover, KID  measures the complexity of an individual sequence, not a stochastic process. The focus of this paper is on stationary processes. Therefore, as the first step, in Section \ref{sec:ID-ergodic}, we develop a measure of complexity, referred to as ID,  for stationary processes.   After that, in Section \ref{sec:universal-cs} we focus on  developing a universal compressed sensing algorithm. In Section \ref{sec:mep},  we propose the MEP optimization, as a universal compressed sensing method, which does not require any prior knowledge about the source model. Then, in Section \ref{sec:theory-mep}, we study the theoretical performance of the MEP. In order to develop a universal compressed sensing algorithm for stochastic processes,  an estimator of their ID   based on the process realizations is required. Therefore, in Section \ref{eq:mixing}, we  study special stationary processes such as $\Psi^*$-mixing processes for which we are able to design such an estimator. Then, for such processes, in Section \ref{sec:mep-perf}, we prove that for the right set of parameters, the normalized number of measurements required by the MEP or by its implementable version,  namely, the Lagrangian-MEP, is slightly more than the ID of the source process.    

 \section{ID of stationary processes}\label{sec:ID-ergodic}

 Consider an   analog random variable $X$ and integer $n\in\mathds{N}$. R\'enyi defined the upper and lower information dimensions of a random variable $X$ in terms of the entropy  of $\langle X \rangle_n $ as
\[
\bar{d}(X)=\limsup\limits_n {H(\langle X \rangle_n) \over \log n},
\]
and
\[
\underline{d}(X)=\liminf\limits_n {H(\langle X \rangle_n) \over \log n},
\]
respectively \cite{Renyi:59}. If $\bar{d}(X)=\underline{d}(X)$, then the information dimension of the random variable $X$ is defined as
\[
d(X)= \lim\limits_{n\to \infty} {H(\langle X \rangle_n) \over \log n}.
\]

 While the R\'enyi  information dimension measure can  be used in measuring the complexity of memoryless continuous-alphabet sources, it cannot directly be applied to  analog stationary processes with memory. For instance, a piecewise-constant signal generated by  a stationary first-order Markov process is expected to be of low complexity. However, the complexity of such processes cannot be evaluated using the   R\'enyi  information dimension or the entropy rate function.  In the following, we develop a measure of complexity for stationary analog sources. To achieve this goal, we carefully combine the definitions of the entropy rate and the R\'enyi information dimension, to capture both the source memory and the fact that the source is continuous-alphabet.

 Define the $b$-bit quantized version of a stochastic process $X=\{X_i\}_{i=1}^{\infty}$ as $[X]_b=\{[X_i]_b\}_{i=1}^{\infty}$. Consider a stationary  process $X=\{X_i\}_{i=1}^{\infty}$;  then since $[X]_b$ is derived from a stationary coding of  $X$, it is also a stationary  process.  We define the $k$-th order upper information dimension of a process $X$ as
\[
\bar{d}_k(X)=\limsup_{b\to \infty} {H([X_{k+1}]_b|[X^k]_b) \over b}.
\]
Similarly, the $k$-th order lower information dimension of  $X$ is defined as
\[
\underline{d}_k(X)=\liminf_{b\to \infty} {H([X_{k+1}]_b|[X^k]_b) \over b}.
\]

\begin{lemma}\label{lemma:1}
Both $\bar{d}_k(X)$ and $\underline{d}_k(X)$  are non-increasing in $k$.
\end{lemma}
\begin{proof}
For  a stationary process $[X]_b$, for any value of $k$,
\begin{align*}
{H([X_{k+2}]_b|[X^{k+1}]_b) \over b}&\leq{H([X_{k+2}]_b|[X^{k+1}_2]_b) \over b}\\
&= {H([X_{k+1}]_b|[X^k]_b) \over b}.
\end{align*}
Therefore, taking $\liminf$ and $\limsup$ of both sides as $b$ grows to infinity  yields the desired result.
\end{proof}

\begin{definition}[Upper/lower information dimension]
For a stationary  process $X$, if $\lim_{k\to\infty}\bar{d}_k(X)$ exists, we define the upper information dimension of process $X$ as
\[
\bar{d}_o(X)=\lim_{k\to\infty}\bar{d}_k(X).
\]
Similarly, if  $\lim_{k\to\infty}\underline{d}_k(X)$ exists, the  lower information dimension of process $X$ is defined as
\[
\underline{d}_o(X)=\lim_{k\to\infty}\underline{d}_k(X).
\]
 If $\underline{d}_o(X) =\bar{d}_o(X)$, ${d}_o(X)\triangleq \underline{d}_o(X) =\bar{d}_o(X)$ is defined as the information dimension of the process $X$.
\end{definition}

\begin{lemma}\label{lemma:1-2}
Consider a  stationary  process $X$, with $\Xc=[l,u]$, where $l<u$ and $l,u\in\mathds{R}$. Then,   $\bar{d}_k(X)\leq 1$ and $\underline{d}_k(X)\leq 1$, for all $k$.
\end{lemma}
\begin{proof}
Note that $H([X_{k+1}]_b|[X^k]_b)\leq \log ((u-l)2^b)=\log(u-l)+b$, for all $b$ and all $k$, and therefore,
\[
{1\over b}H([X_{k+1}]_b|[X^k]_b)\leq 1+{\log(u-l)\over b}.
\]
Taking $\limsup$ and $\liminf$ of both sides as $b$ grows to infinity yields   $\bar{d}_k(X)\leq 1$ and $\underline{d}_k(X)\leq 1$,  for all $k$.
\end{proof}

For  stationary bounded processes,  from Lemmas \ref{lemma:1} and \ref{lemma:1-2},  $\bar{d}_k(X)$ and $\underline{d}_k(X)$ are monotonic bounded sequences. Therefore, $\lim_{k\to\infty}\bar{d}_k(X)$ and $\lim_{k\to\infty}\underline{d}_k(X)$ both exist, and the upper and lower information dimensions of such processes are well-defined.

The entropy rate of a discrete stationary process can be defined either as \eqref{eq:entropy-rate} or equivalently as $\bar{H}(X)=\lim_{k\to\infty}H(X_k|X^{k-1})$.  In the same spirit, the following lemma presents an equivalent  representation of the upper information dimension of a process.

\begin{lemma}\label{lemma:eq-rep}
For  a stationary  process $X$, with upper information dimension $\bar{d}_o(X)$,
\[
\bar{d}_o(X)=\lim_{k\to\infty}{1\over k} \Big(\limsup_{b\to\infty}{H([X^k]_b)\over b}\Big).
\]
\end{lemma}

The following proposition proves that the information dimension of stationary memoryless  processes is equal to the R\'enyi information dimension of their first order marginal distribution. This implies that our notion of  information dimension for stochastic processes is consistent with the R\'enyi's notion of information dimension for random variables.

\begin{proposition}\label{prop:1}
For an i.i.d.~process $X=\{X_i\}_{i=1}^{\infty}$, $\bar{d}_o(X)$ ($\underline{d}_o(X)$) is equal to $\bar{d}(X_1)$ ($\underline{d}(X_1)$), the R\'enyi upper (lower) information dimension of $X_1$.\end{proposition}
\begin{proof}
Since the process is memoryless,  for any quantization level $b$, $H([X_{k+1}]_b|[X^k]_b)=H([X_{k+1}]_b)$ and since it is stationary, $H([X_{k+1}]_b)=H([X_1]_b)$.  Therefore,
\[
\bar{d}_k(X)=\bar{d}_0(X)=\limsup_{b} {H([X_1]_b) \over b}.
\]
As proved in Proposition 2 of  \cite{WuV:10},
\[
\bar{d}(X_1)= \limsup\limits_b {H(\langle X_1\rangle_{2^b})  \over b}.
\]
Since $\langle X_1\rangle_{2^b}=[X_1]_b$, this yields the desired result.
\end{proof}

To clarify the notion of information dimension, in the following we present several examples of different stationary processes and evaluate their information dimensions.

The following theorem, which follows from Theorem 3 of  \cite{Renyi:59} combined with Proposition \ref{prop:1}, characterizes the information dimension of i.i.d~processes, whose components are drawn from a mixture of continuous and discrete distribution.

\begin{theorem}[Theorem 3 in \cite{Renyi:59}]\label{thm:0} Consider an i.i.d.~process $\{X_i\}_{i=1}^{\infty}$, where each $X_i$ is distributed according to
\[
(1-p)f_d+pf_c,
\]
where $f_d$ and $f_c$ represent  a discrete measure and an absolutely continuous measure, respectively. Also, $p\in[0,1]$ denotes the probability that $X_i$ is drawn from the continuous distribution $f_c$. Assume that $H(\lfloor X_1\rfloor)<\infty$. Then,
\[
\bar{d}_o(X)=\underline{d}_o(X)={d}_o(X)=p.
\]
\end{theorem}

From Theorem \ref{thm:0}, for an i.i.d.~process with components drawn from an absolutely continuous distribution\footnote{A probability distribution is called   absolutely continuous, if it has a probability density function (pdf). } the information dimension is equal to one. As a reminder, from Lemma \ref{lemma:1-2}, for    sources with bounded alphabet,   $d_o(X) \leq 1$. Therefore, from Theorem \ref{thm:0}, memoryless sources with an absolutely continuous distribution have maximum complexity. As $p$, the weight of the continuous component,  decreases from one, the information dimension of the source, or equivalently its complexity, decreases as well.

Processes with piecewise constant realizations  are one of the standard models in image processing, and  are  studied in various problems such as denoising  and compressed sensing. Such processes can be modeled as a first-order Markov process.    Theorem \ref{thm:1} evaluates the information dimension of such processes, and shows that their complexity depends on the rate of their jumps. Before that, Theorem \ref{thm:bound-do} connects the information dimension  of a Markov process of order $l$ to its $l$-th order information dimension.

\begin{theorem}\label{thm:bound-do}
Consider a stationary Markov process $X$ of order $\ell$. Then,
\[
\limsup_{b\to\infty} {H([X_{l+1}]_b|X^l) \over b}\leq \bar{d}_o(X)\leq \bar{d}_{\ell}(X),
\]
and
\[
\liminf_{b\to\infty} {H([X_{l+1}]_b|X^l)\over b}\leq  \underline{d}_o(X)\leq \underline{d}_{\ell}(X).
\]
\end{theorem}
\begin{proof}
The upper bounds on both cases follow from Lemma \ref{lemma:1}.
To prove the lower bound, note that for $k>l$,
\begin{align}
H([X_{k+1}]_b|[X^k]_b])&\geq H([X_{k+1}]_b|[X^k]_b],X^{k}_{k-l+1})\nonumber\\
&\stackrel{(a)}{=}H([X_{k+1}]_b|X^{k}_{k-l+1})\nonumber\\
&\stackrel{(b)}{=}H([X_{l+1}]_b|X^l),\label{eq:Hk-Hl-Markov}
\end{align}
where $(a)$ holds because $X$ is a Markov process of order $l$ and therefore $[X^k]_b \to X^{k}_{k-l+1}\to [X_{k+1}]_b$. Equality (b) follows from the stationarity of $X$. Taking $\limsup$ of the both sides of \eqref{eq:Hk-Hl-Markov}, it follows that
\[
\bar{d}_k(X)=\limsup_{b\to\infty} {H([X_{k+1}]_b|[X^k]_b])\over b}\geq \limsup_{b\to\infty}{H([X_{l+1}]_b|X^l)\over b}.
\]
Similarly, taking $\liminf$ of the both sides yields the lower bound on $ \underline{d}_o(X)$.
\end{proof}

\begin{theorem}\label{thm:1}
Consider a first-order stationary Markov process $X=\{X_i\}_{i=1}^{\infty}$, such that  conditioned on $X_{t-1}=x_{t-1}$, $X_t$ has a mixture of discrete and absolutely continuous distribution  equal to   $(1-p) \delta_{x_{t-1}}+ p f_c $, where $f_c$ represents the pdf of an absolutely continuous distribution over $[0,1]$ with bounded differential entropy. Then,
\[
{d}_o(X)=p.
\]
\end{theorem}

As another example, Theorem \ref{thm:2} below considers a special  type of auto-regressive Markov processes of order $l$, $l\in\mathds{N}$, and evaluates their information dimension.

\begin{theorem}\label{thm:2}
Consider a stationary Markov process of order $l$ such that conditioned on $X_{t-l}^{t-1}=x_{t-l}^{t-1}$, $X_t$ is distributed as $\sum_{i=1}^la_ix_{t-i}+Z_t$, where $a_i\in(0,1)$, for $i=1,\ldots,l$, and $Z_t$ is an i.i.d.~process distributed according to $(1-p)\delta_0+pf_c$, where $f_c$ is the pdf of an  absolutely continuous distribution. Let $\Zc$ denote the support of $f_c$ and  assume that there  exists $0<\a<\b<\infty$, such that $\a<f_c(z)<\b$, for $z\in\Zc$. Then,
\[
{d}_o(X)=p.
\]
\end{theorem}

Finally, the last result of this section is concerned with moving average processes, when the original process is a sparse one.

\begin{theorem}\label{thm:moving-ave}
Consider an i.i.d.~sparse process $Y$, such that $Y_i\sim pf_c+(1-p)\delta_0$, where $f_c$ denotes an absolutely continuous distribution with bounded support $(0,1)$. Define the causal moving average of process $Y$ as process $X$ defined as $X_i={1\over l}{\sum_{j=1}^l}Y_{i-j}$. Then,
\[
\bar{d}_o(X)\leq p.
\]
\end{theorem}



\section{Universal CS algorithm}\label{sec:universal-cs}

Consider a stationary process $X=\{X_i\}_{i=1}^{\infty}$, such that $\bar{d}_o(X)< 1$. As we argued in Section \ref{sec:ID-ergodic}, since $\bar{d}_o(X)$ is strictly smaller than one, we expect this process to be  structured. Therefore, intuitively, it might be possible to recover $X_o^n$ generated by source $X$ from an undersampled set of linear measurements $Y_o^m=AX_o^n$, $m<n$. In this section, we explore universal compressed sensing of such processes. We develop algorithms that are able to recover $X_o^n$ from enough  linear measurements,  without having any prior information about the source distribution. The proposed algorithms achieve the optimal performance for stationary memoryless sources with  mixtures of discrete-continuous distributions, and therefore prove that, at least for such memoryless sources, there is no loss in the performance due to universal coding.

 \subsection{Minimum entropy pursuit}\label{sec:mep}


Consider the standard compressed sensing setup:  instead of observing $X_o^n$, the decoder observes  $Y_o^m=AX_o^n$, where $A\in\mathds{R}^{m\times n}$ denotes the linear measurement matrix, and  $m<n$. Further assume that the decoder does not have any knowledge about the distribution of the source.  As we argued, for stationary   processes $\bar{d}_o(X)$ measures the complexity  of the source process. As a reminder,  $\bar{d}_o(X)$ was defined as the limit of  $\bar{d}_k(X)=\limsup_{b}{H([X_{k+1}]_b|[X^{k}]_b)/ b}$. This suggests that, for the right choice of the parameters $b$ and $k$,  $\hat{H}_k([X^n]_b)/b$  defined as $H(U_{k+1}|U^K)$, where $U^{k+1}\sim p_{k+1}(\cdot|[X^n]_b)$, might serve as an estimator of   $\bar{d}_o(X)$. Therefore, inspired by  Occam's razor and this intuition, we propose   \emph{minimum entropy pursuit} (MEP) optimization, which recovers $X_o^n$ by solving the following optimization problem:
\begin{align}
\Xh_o^n=\argmin_{Ax^n=Y_o^m} \hat{H}_k([x^n]_b),\label{eq:MEP-ideal}
\end{align}
where $k$ and $b$ are parameters of the optimization. 

Note that MEP does not require knowledge of  the source  distribution and  hence is a universal compressed sensing recovery algorithm. In the next section, we prove that if the stationary process  $X$ satisfies certain constraints that will be specified later, and the parameters $k$ and $b$ are set appropriately,  then the MEP optimization  can reliably recover $X_o^n$ from measurements $Y^m$, as long as $m> (1+\delta) \bar{d}_o(X)n$, where $\delta>0$ can get arbitrary small.

 \subsection{Theoretical analysis of  MEP}\label{sec:theory-mep}
The main goal of this section  is to show that MEP succeeds in recovering the source vector, without having access to its  distribution. However, to prove this result, we have to impose a constraint on the source distribution. We first review this condition  in Section \ref{eq:mixing}, and then under these constraints, we characterize  the performance of MEP in Section \ref{sec:mep-perf}.  

\subsubsection{Mixing processes}\label{eq:mixing}
To gain some insight on the constraint imposed on the input process, consider a simpler question. Suppose that the decoder has access to both the noiseless measurements $Y^m=AX_o^n$, and the complexity, or more specifically,  the upper ID, of the  process that has generated $X_o^n$. Now given a potential reconstruction sequence $\hat{x}^n$,  is it possible to confirm whether $\xh^n$ is a solution of the MEP optimization? Our heuristic answer to this question, may help the reader understand the constraints studied later.  It is straightforward to check whether $\hat{x}^n$ satisfies the measurement constraints. Moreover, for $\xh^n$ to be a solution of the MEP, in addition to satisfying the measurement equations,   its complexity is expected to be close to the complexity of sequences generated by the source. 
For instance, if the decoder has access to  a reliable estimator of  $\bar{d}_o(X)$, it can apply it to  $\xh^n$, and for $n$ large enough, if $\xh^n$ is equal or close to the source vector, it can expect   the output to be close to $\bar{d}_o(X)$. The constraints imposed on process  $X$ enables us to develop such sn estimator of $\bar{d}_o(X)$.

In the rest of the paper, we focus on  $\psi^*$-mixing  stationary  processes. This condition  ensures  the  convergence of our estimate of $\bar{d}_o(X)$. The $\psi^*$-mixing condition  is a standard property studied in the ergodic theory literature. 
Consider a stationary  process $X=\{X_n\}_{-\infty}^\infty$. Let $\mathcal{F}_j^\ell$ denote the $\sigma$-field of events generated by random variables $X_j^k$, where $j\leq k$. Define 
\begin{equation}
\psi^*(g) = \sup \frac{\P(\Ac \cap \Bc)}{\P(\Ac) \P(\Bc)},
\end{equation}
where the supremum  is taken over all events $\Ac \in \mathcal{F}_{-\infty}^{j}$ and $\Bc \in \mathcal{F}_{j+g}^\infty$, where $P(\Ac)>0$ and $\P(\Bc)>0$.
\begin{definition} A  process $X=\{X_n\}_{-\infty}^\infty$ is called $\psi^*$-mixing, if $\psi^*(g) \to 1$, as $g$ grows to infinity.
\end{definition}

Intuitively, this condition ensures that the future and the past of the process that are well-separated are almost independent from each other. (For more information on $\psi^*$-mixing condition, and its connection to other mixing conditions,  the reader is referred to \cite{bradley2005basic}.)
 In the following we review some $\psi^*$-mixing processes.
\begin{example}
Any i.i.d.~process is $\psi^*$-mixing. 
\end{example}

\begin{example}
All aperiodic  Markov chains with finite sate space are  $\psi^*$-mixing \cite{Shields:96}.
\end{example}
\begin{example}
Consider an i.i.d. process $Y=\{Y_i\}$ and define its moving average as $X_i={1\over l}\sum_{j=1}^lY_{i-j}$. Then, process $X$ is $\Psi^*$-mixing. 
\end{example}
\begin{proof}
Note that since process $Y$ is i.i.d.~and since $l$ is finite, for $g$ large enough, $\Fc_{-\infty}^{j}$ and $\Fc_{j+g}^{\infty}$ are independent and therefore $\Psi^*(g)=1.$
\end{proof}
In fact, by the same proof, the averaging function can be replaced by any fixed mapping  $f: \Xc_{-l}^l\to\Xc$ and the process defined as $X_i=f(Y_{i-l}^{i+l})$  is still $\Psi^*$-mixing.

Theorem III.1.7 of \cite{Shields:96} proves that any $\Psi^*$ mixing process with finite alphabet  has exponential rates of convergence  for empirical distributions of all orders. The following theorem presents a straightforward extension of that result to $\Psi^*$-mixing processes with continuous alphabet. It proves that  $b$-bit quantized versions of such processes  have exponential rates for empirical frequencies of all orders, if  $b$ is growing with $n$ slowly enough. 
\begin{theorem}\label{thm:exp-rate-markov}
Consider   $\Psi^*$-mixing process $X=\{X_i\}$, with continuous alphabet $\Xc$. Let process $Z$ denote the  $b$-bit quantized version of process $X$. That is, $Z=\{Z_i\}$,  $Z_i=[X_i]_b$  and $\Zc=\Xc_b$. Then, for any $\e>0$,  there exists $g\in\mathds{N}$,  depending only on $\e$, such that for any  $n>6(k+g)/\e+k$,
\[
\P( \|p_k(\cdot|Z^n)-\mu_k\|_1\geq \e)\leq2^{c\e^2/8} (k+g) n^{|\Zc|^k}2^{-nc\e^2\over 8(k+g)},
\]
where $c=1/(2\ln 2)$. Here, for $a^k\in\Zc^k,$  $\mu_k(a^k)=\P(Z^k=a^k)$
\end{theorem}
\begin{proof}
The proof is excatly as the proof of Theorem  III.1.7 of \cite{Shields:96}. The shift from continuous-alphabet sources to finite-alphabet sources is done by the quantization of the source and by also noting that if a continuous-alphabet  process is $\Psi^*$-mixing, its quantized version is also $\Psi^*$-mixing. To see this, for $j\leq k$, let $\Fc_{j}^k$ and $\hat{\Fc}_{j}^k$ denote the sigma-fields generated by $X_{j}^k$ and $Z_{j}^k$, respectively. But since $\hat{\Fc}_{j}^k$ is always a sub sigma-field of ${\Fc}_{j}^k$, we have
\begin{align*}
\psi_Z^*(g) &= \sup_{\Ac\in\hat{\Fc}_{-\infty}^j,\Bc\in\hat{\Fc}_{j+g}^{\infty}} \frac{\P(\Ac \cap \Bc)}{\P(\Ac) \P(\Bc)}\\
&\leq \sup_{\Ac\in{\Fc}_{-\infty}^j,\Bc\in{\Fc}_{j+g}^{\infty}} \frac{\P(\Ac \cap \Bc)}{\P(\Ac) \P(\Bc)}\\
& =\Psi_X^*(g).
\end{align*}
But since the continuous process is known to be $\Psi^*$-mixing, $\Psi_X^*(g)$ converges to one, as $g$ grows to infinity. This proves that process $Z$ is also $\Psi^*$-mixing, with a $\Psi^*(g)$ function than is upper-bounded with that of $\Psi_X^*(g)$.
\end{proof}

\subsubsection{Performance of MEP}\label{sec:mep-perf}
The following theorem proves that MEP is a universal decoder for $\Psi^*$-mixing processes. 
\begin{theorem}\label{thm:3}
Consider a $\Psi^*$-mixing  stationary  process $\{X_i\}_{i=1}^{\infty}$, with $\Xc=[0,1]$ and  upper information dimension $\bar{d}_o(X)$. Let $b=b_n=\lceil \log\log n\rceil $, $k=k_n=o({\log n\over \log \log n})$ and $m=m_n\geq (1+\d)\bar{d}_o(X)n$, where $\d>0$. For each $n$, let the entries of the measurement matrix  $A=A_n\in\mathds{R}^{m\times n}$ be drawn  i.i.d.~according to $\Nc(0,1)$. For $X_o^n$ generated by the source $X$ and $Y_o^{m}=AX_o^n$, let $\Xh_o^n=\Xh^n_{o}(Y_o^{m},A)$ denote the solution of \eqref{eq:MEP-ideal}, \ie $\Xh_o^n=\argmin_{Ax^n=Y_o^{m}} \hat{H}_{k}([x^n]_{b})$. Then,
\[
{1\over \sqrt{n}}\|X_o^n-\Xh_o^n\|_2\stackrel{\rm P}{\longrightarrow} 0.
\]
\end{theorem}

\begin{remark}
Theorem \ref{thm:3} proves that, in the asymptotic setting, as the blocklength $n$ grows to infinity, the normalized number of measurements ($m/n$) required by  MEP  for recovering memoryless sources with discrete-continuous mixture distributions coincides with  the fundamental limits of  non-universal compressed sensing characterized in \cite{WuV:10}. In other words,  this result shows that, at least for such sources,  there is no loss in  performance due to universality.  This proves that, asymptotically, at least for such stationary  memoryless sources, similar to data compression, denoising, and prediction, there is no loss in universal compressed sensing, due to not knowing the source distribution.
\end{remark}

The optimization presented in \eqref{eq:MEP-ideal} is not easy to handle. While the search domain, \ie the set of points satisfying $Ax^n=y_o^m$, is a hyperplane, the cost function is defined on a discretized space, which is formed by the quantized version of the source alphabet $\Xc$. To move towards  designing an implementable universal compressed sensing algorithm, consider the following Lagrangian-type approximation of MEP:
 \begin{align}
\xh_o^n=\argmin_{u^n\in\Xc_b^n} \Big(\hat{H}_k(u^n)+{\lambda\over n^2} \|Au^n-y_o^m\|_2^2\Big),\label{eq:MEP-lagrangian}
\end{align}
where $\Xc_b\triangleq \{[x]_b:\;x\in\Xc\}$.
We refer to this algorithm as Lagrangian-MEP. The main difference between \eqref{eq:MEP-ideal} and \eqref{eq:MEP-lagrangian} is that in \eqref{eq:MEP-lagrangian} the search space is now a discrete set. The advantage of  Lagrangian-MEP compared to MEP is  that it is implementable and  classic discrete optimization methods such as Markov chain Monte Carlo (MCMC)  and simulated annealing  \cite{Gibbs_sampler,simulated_annealing_1,simulated_annealing_2} can be employed to approximate its optimizer.

The Lagrangian-MEP algorithm  is  in fact identical to the heuristic algorithm for universal compressed sensing proposed in \cite{BaronD:11} and \cite{BaronD:12}.   In \cite{BaronD:11} and \cite{BaronD:12}, Baron et al.~ employ simulated annealing and Markov chain Monte Carlo techniques to approximate  the minimizer of the Lagrangian-MEP cost function.    The following theorem shows that for the right choice of parameter $\lambda$, \eqref{eq:MEP-lagrangian} is in fact a universal compressed sensing algorithm, which approximates the solution of MEP with no asymptotic loss in    performance.

\begin{theorem}\label{thm:4}
Consider a  $\Psi^*$-mixing  stationary   process $\{X_i\}_{i=1}^{\infty}$, with $\Xc=[0,1]$ and  upper information dimension $\bar{d}_o(X)$. Let $b=b_n=\lceil r\log\log n\rceil $, where $r>1$, $k=k_n=o({\log n\over \log \log n})$, $\l=\l_n=(\log n)^{2r}$ and $m=m_n\geq (1+\d)\bar{d}_o(X)n$, where $\d>0$. For each $n$, let the entries of the measurement matrix  $A=A_n\in\mathds{R}^{m\times n}$ be drawn i.i.d.~according to $\Nc(0,1)$. Given $X_o^n$ generated by source $X$ and $Y_o^{m}=AX_o^n$, let $\Xh_o^n=\Xh^n_{o}(Y_o^{m},A)$ denote the solution of \eqref{eq:MEP-lagrangian}, \ie $\Xh_o^n=\argmin_{u^n\in\Xc^n} \;(\hat{H}_{k}(u^n)+{\lambda \over n^2} \|Au^n-Y_o^m\|_2^2)$. Then,
\[
{1\over \sqrt{n}}\|X_o^n-\Xh_o^n\|_2\stackrel{\rm P}{\longrightarrow} 0.
\]
\end{theorem}



  So far we assumed that the measurements are perfect and noise-free. In almost all practical situations the measurements are contaminated by noise. Therefore, it is important to study the performance of the proposed algorithms in the presence of noise. We next prove that the Lagrangian-MEP algorithm, which was proved to be  an implementable universal compressed sensing algorithm, is also robust to measurement noise.

 Assume that  instead of $Ax_o^n$, the decoder observes  $y_o^m=Ax_o^n+z^m$, where $z^m$ denotes the noise in the measurement system, and   employs the Lagrangian-MEP to recover $x_o^n$, \ie
\[
\xh^n=\argmin_{ u^n\in\Xc_b^n}\Big(\hat{H}_k(u^n)+{\lambda\over n^n}\|Au^n-y_o^m\|_2\Big).
\]

The following theorem proves that the Lagrangian-MEP is robust to measurement noise, and as long as the $\ell_2$ norm of the noise vector is small enough, the algorithm recovers the source vector from the same number of measurements, despite receiving noisy observations.

\begin{theorem}\label{thm:5}
Consider a  $\Psi^*$-mixing  stationary  process $\{X_i\}_{i=1}^{\infty}$, with $\Xc=[0,1]$ and  upper information dimension $\bar{d}_o(X)$. Consider a measurement matrix  $A=A_n\in\mathds{R}^{m\times n}$ with i.i.d.~entries distributed according to $\Nc(0,1)$. Let $b=b_n=\lceil r\log\log n\rceil $, where $r>1$, $k=k_n=o({\log n\over \log \log n})$, $\l=\l_n=(\log n)^{2r}$ and $m=m_n\geq (1+\d)\bar{d}_o(X)n$, where $\d>0$.  For $X_o^n$ generated by the source $X$, we observe $Y_o^{m}=AX_o^n+Z^m$, where $Z^m$ denotes the measurement noise. Assume that there exists a deterministic sequence $c_m$ such that  $\lim\limits_{m\to\infty}\P(\|Z^m\|_2>c_m)=0$, and  $c_m=O({m/ (\log m)^r})$.
Let $\Xh_o^n=\Xh^n_{o}(Y_o^{m},A)$ denote the solution of \eqref{eq:MEP-lagrangian}. Then, ${1\over \sqrt{n}}\|X_o^n-\Xh_o^n\|_2\stackrel{\rm P}{\longrightarrow}0.$
\end{theorem}

%


\section{Proofs}\label{sec:proof}

Before presenting the proofs of the results, we state two useful lemmas that are used later in the proofs.  The first  lemma in the following is from \cite{JalaliM:14}.

\begin{lemma}[$\chi^2$ concentration]\label{lemma:chi}
Fix $\tau>0$, and let $U_i\stackrel{\rm i.i.d.}{\sim}\Nc(0,1)$, $i=1,2,\ldots,m$. Then,
\begin{align*}
\P\Big( \sum_{i=1}^m  U_i^2 <m(1- \tau) \Big)  \leq {\rm e} ^{\frac{m}{2}(\tau + \ln(1- \tau))}
\end{align*}
and
\begin{align}\label{eq:chisq}
\P\Big( \sum_{i=1}^m  U_i^2 > m(1+\tau) \Big)  \leq {\rm e} ^{-\frac{m}{2}(\tau - \ln(1+ \tau))}.
\end{align}
\end{lemma}

%

\begin{lemma}\label{lemma:bd-TV-dist}
Consider distributions $p$ and $q$ on finite alphabet $\Xc$ such that $\|p-q\|_1\leq \e$. Then,
\[
|H(p)-H(q)|\leq -\e\log \e +\e\log |\Xc|.
\]
\end{lemma}
\begin{proof}
Define $f(y)=-y\ln y$, for $y\in[0,1]$, and $g(y)=f(y+\e)-f(y)$. Since $g'(y)=\ln(y/(y+\e))<0$, $g$ is a decreasing function of $y$. Therefore,
\[
g(y)\leq -\e\ln\e.
\]
For $x\in\Xc$, let $|p(x)-q(x)|=\e_x$. By our assumption,
\begin{align}
\sigma\triangleq \sum_{x\in\Xc}\e_x\leq \e.\label{eq:e-x}
\end{align}
On the other hand, we just proved that
\[
|-p(x)\ln p(x)+q(x)\ln q(x)| \leq -\e_x\ln \e_x.
\]
Therefore,
\begin{align}
|H(p)-H(q)|&=|\sum_{x\in\Xc}(-p(x)\ln p(x)+q(x)\ln q(x))|\nonumber\\
&\leq \sum_{x\in\Xc}|-p(x)\ln p(x)+q(x)\ln q(x)|\nonumber\\
&\leq \sum_{x\in\Xc}  -\e_x\ln \e_x.
\end{align}
Also
\begin{align}
\sum_{x\in\Xc}  -\e_x\log \e_x &=\sum_{x\in\Xc}  -\e_x\log  {\e_x \sigma\over \sigma}\nonumber\\
 &=-\sigma\log \sigma +\sigma H({\e_x\over \sigma}: x\in\Xc)\nonumber\\
  &\leq -\e\log \e +\e \log |\Xc|,
\end{align}
where the last line follows because $f(y)$ is an increasing function for $y\leq \ex^{-1}$.
\end{proof}

\subsection{Proof of Lemma \ref{lemma:eq-rep}}

Since the process is stationary,
\begin{align}
H([X^k]_b)&=\sum_{i=1}^kH([X_i]_b|[X^{i-1}]_b)\nonumber\\
&=\sum_{i=1}^kH([X_k]_b|[X^{k-1}_{k-i+1}]_b)\nonumber\\
&\geq kH([X_k]_b|[X^{k-1}]_b).
\end{align}
Therefore,
\begin{align*}
{1\over k} \Big(\limsup_{b\to\infty}{H([X^k]_b)\over b}\Big)&\geq \limsup_{b\to\infty}{H([X_k]_b|[X^{k-1}]_b)\over b}\\
&=\bar{d}_k(X).
\end{align*}
Taking $\liminf$ of both as $k$ grows to infinity proves that
\begin{align}
\liminf_{k\to\infty}{1\over k} \Big(\limsup_{b\to\infty}{H([X^k]_b)\over b}\Big)\geq \bar{d}_o(X).\label{eq:lb-d-k}
\end{align}
 On the other hand, for any set of functions $f_1,\ldots,f_k$, and any $B\in\mathds{R}$, $\sup_{b>b_o}\sum_{i=1}^kf_i(b)\leq \sum_{i=1}^k \sup_{b>b_o}f_i(b)$. Taking the limit of both sides as $b_o$ grows to infinity yields  $\limsup_{b}\sum_{i=1}^kf_i(b)\leq \sum_{i=1}^k \limsup_{b}f_i(b)$. Therefore, letting $f_i(b)=b^{-1}H([X_k]_b|[X_{k-i+1}^{k-1}]_b)$, it follows that
 \begin{align}
\limsup_{b\to\infty}{H([X^k]_b)\over bk}&=\limsup_{b\to\infty} {\sum_{i=1}^kH([X_k]_b|[X^{k-1}_{k-i+1}]_b)\over bk}\nonumber\\
&\leq {1\over k}\sum_{i=1}^k \limsup_{b\to\infty} {H([X_k]_b|[X^{k-1}_{k-i+1}]_b)\over b}\nonumber\\
&= {1\over k}\sum_{i=0}^{k-1} \bar{d}_i(X).\label{eq:limsup-sum}
\end{align}
For a sequence of numbers $(a_k)_k$, such that $\lim_{k\to\infty}a_k=a$, the Ces\`aro mean theorem states that $b_k={1\over k}\sum_{i=1}^ka_i$ also converges to $a$. Therefore,  since $ \lim_{i\to\infty}\bar{d}_i(X)=\bar{d}_o(X)$, by the  Ces\`aro mean theorem, we have
\[
\lim_{k\to\infty} {1\over k}\sum_{i=0}^{k-1} \bar{d}_i(X)=\bar{d}_o(X).
\]
Taking the $\limsup$ of both sides of \eqref{eq:limsup-sum}  yields
\begin{align}
\limsup_{k\to\infty}\limsup_{b\to\infty}{H([X^k]_b)\over bk}\leq \bar{d}_o(X).\label{eq:ub-d-k}
\end{align}
The desired result follows from combining \eqref{eq:lb-d-k} and \eqref{eq:ub-d-k}.


\subsection{Proof of Theorem \ref{thm:1}}\label{app:A}

We first show that the quantized versions of $X$ at any quantization level $b$ is also a stationary first-order Markov process. To show this,  let $Z_k$ denote the i.i.d.~Bernoulli process that indicates the positions of the jumps in process $X$. In other words, $Z_k=\ind_{X_k\neq X_{k-1}}.$ Then, for any $u^{k+1}\in\Xc_b^{k+1},$ we have
\begin{align*}
\P([X_{k+1}]_b=u_{k+1}|[X^k]_b=u^k)=&\P([X_{k+1}]_b=u_{k+1},Z_{k+1}=0|[X^k]_b=u^k)\\
&+\P([X_{k+1}]_b=u_{k+1},Z_{k+1}=1|[X^k]_b=u^k).
\end{align*}
But,  since by the definition of process $X$ , $Z_{k+1}$ is independent of $X^{k}$, 
\begin{align*}
\P([X_{k+1}]_b=u_{k+1},Z_{k+1}=0|[X^k]_b=u^k)&=\P(Z_{k+1}=0|[X^k]_b=u^k)\P([X_{k+1}]_b=u_{k+1}|Z_{k+1}=0,[X^k]_b=u^k)\\
&=(1-p)\P([X_{k+1}]_b=u_{k+1}).
\end{align*}
and 
\begin{align*}
\P([X_{k+1}]_b=u_{k+1},Z_{k+1}=1|[X^k]_b=u^k)&=\P([X_{k+1}]_b=\P(Z_{k+1}=1|[X^k]_b=u^k)u_{k+1}|Z_{k+1}=1,[X^k]_b=u^k)\\
&=p\ind_{u_{k+1}=u_k}.
\end{align*}
Therefore, overall,
\begin{align*}
\P([X_{k+1}]_b=u_{k+1}|[X^k]_b=u^k)=&(1-p)\P([X_{k+1}]_b=u_{k+1})+p\ind_{u_{k+1}=u_k},
\end{align*}
which only depends on $u_k$. Therefore $[X]_b$ is also a first-order Markov process. Stationarity of $[X]_b$ follows immediately. 

Now since the quantized  process is also a stationary first-order Markov process,
\[
H([X_{k+1}]_b|[X^k]_b)=H([X_{k+1}]_b|[X_k]_b)=H([X_2]_b|[X_1]_b).
\]
Therefore,
\[
\bar{d}_k(X)=\bar{d}_1(X),
\]
and
\[
\underline{d}_k(X)=\underline{d}_1(X),
\]
for all $k\geq 1$.
Let $\Xc_b=\{[x]_b: x\in\Xc \}$ denote the alphabet at resolution $b$.  Then,
\[
H([X_2]_b|[X_1]_b)=\sum_{c\in\Xc_b}\P([X_1]_b=c) H([X_2]_b|[X_1]_b=c).
\]
We next prove that ${1\over b}H([X_2]_b|[X_1]_b=c_1)$ uniformly converges to $p$, as $b$ grows to infinity, for all values of $c_1$.

Define the indicator random variable $I=\ind_{X_2=X_1}$. Given the transition probability of the Markov chain,  $I$ is independent of $X_1$, and $\P(I=1)=1-p$. Also define a random variable $U$, independent of $(X_1,X_2)$, and   distributed according to $f_c$. Then, it follows that
\begin{align*}
\P([X_2]_b=c_2|[X_1]_b=c_1)&=\P([X_2]_b=c_2,I=0|[X_1]_b=c_1)\nonumber\\
&\;\;\;+\P([X_2]_b=c_2,I=1|[X_1]_b=c_1)\nonumber\\
&=p\P([X_2]_b=c_2|[X_1]_b=c_1,I=0)\nonumber\\
&\;\;\;+(1-p)\P([X_2]_b=c_2|[X_1]_b=c_1,I=1)\nonumber\\
&=p\P([U]_b=c_2)+(1-p)\ind_{c_2=c_1},
\end{align*}
where the last line follows from the fact that conditioned on $X_2\neq X_1$, $X_2$, independent of the value of $X_1$, is distributed according to $f_c$. For $a\in\Xc_b$,
\begin{align}
P([U]_b=a)=\int_{a}^{a+2^{-b}}f_c(u)du.\label{eq:P-U-b-mean-value-thm}
\end{align}
On the other hand, by the mean value theorem, there exists $x_a\in[{a},{a+2^{-b}}]$, such that
\[
2^b\int_{a}^{a+2^{-b}}f_c(u)du=f_c(x_a).
\]
Therefore, $P([U]_b=a)=2^{-b}f_c(x_a)$, and
\begin{align}
H([X_2]_b|[X_1]_b=c_1)&=-\sum_{a\in\Xc_b, a\neq c_1 }- (p2^{-b}f_c(x_a))\log (p2^{-b}f_c(x_a))\nonumber\\
&\;\;\;\;\;\;-(p2^{-b}f_c(x_{c_1})+1-p)\log (p2^{-b}f_c(x_{c_1})+1-p)\nonumber\\
&= \sum_{a\in\Xc_b, a\neq c_1 }- (p2^{-b}f_c(x_a))(\log p -b +\log (f_c(x_a)))\nonumber\\
&\;\;\;\;\;\;-(p2^{-b}f_c(x_{c_1})+1-p)\log (p2^{-b}f_c(x_{c_1})+1-p).\label{eq:H-X2-given-X1}
\end{align}
Dividing both sides of \eqref{eq:H-X2-given-X1} by $b$ yields
\begin{align}
{H([X_2]_b|[X_1]_b=c_1)\over b}
&= {(b-\log p)p \over b}\sum_{a\in\Xc_b, a\neq c_1 } 2^{-b}f_c(x_a) \nonumber\\
&\;\;\;\;\;\;-({p\over b})\sum_{a\in\Xc_b, a\neq c_1 } 2^{-b}f_c(x_a) \log (f_c(x_a))\nonumber\\
&\;\;\;\;\;\;-{(p2^{-b}f_c(x_{c_1})+1-p)\log (p2^{-b}f_c(x_{c_1})+1-p)\over b}.\label{eq:H-X2-given-X1-over-b}
\end{align}
On the other hand, from \eqref{eq:P-U-b-mean-value-thm},
\begin{align}
\sum_{a\in\Xc_b } 2^{-b}f_c(x_a)&=\sum_{a\in\Xc_b}\int_{a}^{a+2^{-b}}f_c(u)du\nonumber\\
&=\int_0^1f_c(u)du\nonumber\\
& =1.\label{eq:sum-eq-1}
\end{align}
Also, since $\int_{0}^1 f_c(u)du=1$,
\begin{align*}
\lim_{b\to\infty}\sum_{a\in\Xc_b } 2^{-b}f_c(x_a) \log (f_c(x_a))=h(f_c),\label{eq:ub-q-log-q}
\end{align*}
where $h(f_c)=-\int f_c(u)\log f_c(u)d u$ denotes the differential entropy of $U$. Therefore, for any $\e>0$, there exists $b_{\e}\in\mathds{N}$, such that for $b>b_{\e}$,
\begin{align*}
|\sum_{a\in\Xc_b } 2^{-b}f_c(x_a) \log (f_c(x_a))-h(f_c)|\leq \e.
\end{align*}
Since $f_c$ is bounded by assumption, $M=\sup_{x\in[0,1]} f_c(x)<\infty$, and $h(f_c)\leq \log M<\infty$. Finally, $-q\log  q\leq \ex^{-1}\log \ex$, for $q\in[0,1]$. Therefore, combining   \eqref{eq:H-X2-given-X1-over-b}, \eqref{eq:sum-eq-1} and \eqref{eq:ub-q-log-q}, it follows that, for $b>b_{\e}$,
\begin{align}
\left|{H([X_2]_b|[X_1]_b=c_1)\over b}-p\right| \leq {M\over 2^b}-{p\log p \over b}+{p(h(f_c)+\e)\over b}+{\log \ex\over \ex b},
\end{align}
for all $c_1\in\Xc_b$. Since the right hand side of the above equation does not depend on $c_1$, and goes to zero as $b\to \infty$, for any $\e'>0$, there exists $b_{\e'}$, such that for $b>\max\{b_{\e},b_{\e'}\}$,
\begin{align}
\left|{H([X_2]_b|[X_1]_b=c_1)\over b}-p\right| \leq \e',
\end{align}
and
\begin{align}
\left|{H([X_2]_b|[X_1]_b)\over b}-p \right|&\leq \sum_{c_1\in\Xc_b}\P([X_1]_b=c_1)\left|{H([X_2]_b|[X_1]_b=c_1) \over b}-p\right|\nonumber\\
&\leq \e'\sum_{c_1\in\Xc_b}\P([X_1]_b=c_1)\nonumber\\
&=\e',
\end{align}
which concludes the proof.


\subsection{Proof of Theorem \ref{thm:2}}\label{app:B}
Since the process is stationary and Markov of order $l$, by Theorem \ref{thm:bound-do} $\limsup_{b\to\infty} b^{-1}H([X_{l+1}]_b|X^l)\leq \bar{d}_o(X)\leq \bar{d}_{\ell}(X)$ and $\liminf_{b\to\infty} b^{-1}H([X_{l+1}]_b|X^l)\leq  \underline{d}_o(X)\leq \underline{d}_{\ell}(X).$

Define the  indicator random variable $I=\ind_{Z_l=0}$. By the definition of the Markov chain, $\P(I=1)=1-p$. Then,  since $I$ is independent of $X^l$, we have
\begin{align}
H([X_{l+1}]_b|[X^l]_b)\;\leq\; & H([X_{l+1}]_b,I|[X^l]_b)\nonumber\\
\;\leq\; & 1+ H([X_{l+1}]_b|[X^l]_b,I)\nonumber\\
\;=\;&1+ p H([\sum_{i=1}^la_iX_{l-i}+U]_b|[X^l]_b)\nonumber\\
&+(1-p)H([\sum_{i=1}^la_iX_{l-i}]_b|[X^l]_b),\label{eq:lb-cond-prob}
\end{align}
where $U$ is independent of $X^l$ and is distributed according to $f_c$. 

Conditioned on $[X^l]_b=c^l$, where  $c_1,\ldots,c_{l+1}\in\Xc_b$, we have
\[
c_i\leq X_i<c_i+2^{-b},
\]
for $i=1,\ldots,l$, and
\[
\Big| \sum_{i=1}^{l}a_i X_{l-i}-\sum_{i=1}^la_ic_{l-i}\Big| \leq 2^{-b}\sum_{i=1}^l|a_i|.
\]
Let $M=\left\lceil \sum_{i=1}^l|a_i| \right\rceil$ and $c=\sum_{i=1}^la_ic_{l-i}$. Then,
\[
\sum_{i=1}^{l}a_i X_{l-i}\in [c-M2^{-b}, c+M2^{-b}].
\]
Therefore, $[\sum_{i=1}^{l}a_i X_{l-i}]_b$ can take only $2M+1$ different values, and as a result $H([\sum_{i=1}^la_iX_{l-i}]_b|[X^l]_b)\leq \log(2M+1)$. Since $M$ does not depend on $b$, it follows that
\[
\lim_{b\to\infty} {H([\sum_{i=1}^la_iX_{l-i}]_b|[X^l]_b)\over b}=0.
\]
We next prove that
\[
\lim_{b\to\infty}{H([\sum_{i=1}^la_iX_{l-i}+U]_b|[X^l]_b)\over b}=1,
\]
for any absolutely continuous distribution  with pdf $f_c$. This proves that $\bar{d}_{l}(X)=\underline{d}_{l}(X)=p.$

To bound $H([\sum_{i=1}^la_iX_{l-i}+U]_b|[X^l]_b)$, we need to study $\P([\sum_{i=1}^la_iX_{l-i}+U]_b=c|[X^l]_b=c^l)$,  where $c^l\in\Xc_b^l$, and $c\in\Xc_b$. Let $\Nc(c^l,b)=\{x^l: c_i\leq x_i\leq c_i+2^{-b}, i=1,\ldots,l\}$, and define the function $g:\mathds{R}^l\to\mathds{R}$, by $g(x^l)=\sum_{i=1}^la_ix_{l-i}$. Note that
\begin{align}
\P\Big([\sum_{i=1}^la_iX_{l-i}+U]_b=c\Big|[X^l]_b=c^l\Big) \;=\; & {\P([\sum_{i=1}^la_iX_{l-i}+U]_b,[X^l]_b=c^l)  \over \P([X^l]_b=c^l) }\nonumber\\
\;=\; & {\int_{\Nc(c^l,b)}\int_{c-g(x^l)}^{c-g(x^l)+2^{-b}} f(x^l) f_c(u)du d x^l  \over \P([X^l]_b=c^l) },\label{eq:bayes-rule}
\end{align}
where $f(x^l)$ denotes the pdf of $X^l$. By the mean value theorem, there exists $\delta(x^l)\in(0,2^{-b})$, such that
\begin{align}
\int_{c-g(x^l)}^{c-g(x^l)+2^{-b}}  f_c(u)du=2^{-b}f_c(c-g(x^l)+\delta(x^l)).\label{eq:mean-value-thm}
\end{align}
Combining \eqref{eq:bayes-rule} and \eqref{eq:mean-value-thm} yields that
\begin{align}
\P\Big([\sum_{i=1}^la_iX_{l-i}+U]_b=c\Big|[X^l]_b=c^l\Big)
\;=\; &  {2^{-b}\int_{\Nc(c^l,b)} f(x^l) f_c(c-g(x^l)+\delta(x^l))d x^l  \over \int_{\Nc(c^l,b)} f(x^l) d x^l }.
\end{align}
Define  the pdf $p_{c^l,b}(y^l)$ over $\Nc(c^l,b)$ as
\[
p_{c^l,b}(y^l)= { f(y^l) \over \int_{\Nc(c^l,b)} f(x^l) d x^l }.
\]
Then, $\P([\sum_{i=1}^la_iX_{l-i}+U]_b=c|[X^l]_b=c^l) =2^{-b}\E[f_c(c-g(Y^l)-\delta(Y^l))]$, where $Y^l\sim p_{c^l,b}$. Hence,
\begin{align}
H\Big([\sum_{i=1}^la_iX_{l-i}+U]_b\Big|[X^l]_b\Big)\;=\;&\sum_{c^l}\sum_c (b-\log \E[f_c(c-g(Y^l)-\delta(Y^l))] )\nonumber\\
&\hspace{1cm}\times\P\Big(\Big[\sum_{i=1}^la_iX_{l-i}+U\Big]_b=c\Big|[X^l]_b=c^l\Big) \nonumber\\
\;=\;&b-\sum_{c^l}\sum_c \log \E[f_c(c-g(Y^l)-\delta(Y^l))] \nonumber\\
&\hspace{1cm}\times\P\Big(\Big[\sum_{i=1}^la_iX_{l-i}+U\Big]_b=c\Big|[X^l]_b=c^l\Big).
\end{align}
Since by assumption $f_c$ is bounded on its support between $\a$ and $\b$, then for $b$ large enough, if $c-\sum_{i=1}^l a_ic_{l-i}\in\Zc$, then $\E[f_c(c-g(Y^l)-\delta(Y^l))]$ is also bounded between $\a$ and $\b$, and hence the desired result follows. That is, $\lim\limits_{b\to\infty} b^{-1}H([\sum_{i=1}^la_iX_{l-i}+U]_b|[X^l]_b)=1$.

For the lower bound, we next prove that $\limsup_{b\to\infty} {H([X_{l+1}]_b|X^l)\over b}\geq p$ and $\liminf_{b\to\infty} {H([X_{l+1}]_b|X^l)\over b}\geq p$.
Note that 
\begin{align}
H([X_{l+1}]_b|X^l)\;\geq\; & H([X_{l+1}]_b|X^l,I)\nonumber\\
\;=\;& p H([\sum_{i=1}^la_iX_{l-i}+U]_b|X^l)+(1-p)H([\sum_{i=1}^la_iX_{l-i}]_b|X^l)\nonumber\\
\;=\;& p H([\sum_{i=1}^la_iX_{l-i}+U]_b|X^l), \label{eq:ub-cond-prob}
\end{align}
where the line follows from the fact that $H([\sum_{i=1}^la_iX_{l-i}]_b|X^l)=0$. Therefore, $\limsup_{b\to\infty} {H([X_{l+1}]_b|X^l)\over b}\geq \limsup_{b\to\infty} {H([\sum_{i=1}^la_iX_{l-i}+U]_b|X^l)\over b}$ and $\liminf_{b\to\infty} {H([X_{l+1}]_b|X^l)\over b}\geq \liminf_{b\to\infty} {H([\sum_{i=1}^la_iX_{l-i}+U]_b|X^l)\over b}$. But, 
\begin{align}
\liminf_{b\to\infty} {H([\sum_{i=1}^la_iX_{l-i}+U]_b|X^l)\over b}&=\liminf_{b\to\infty} \int{1\over b}H([\sum_{i=1}^la_ix_{l-i}+U]_b|X^l=x^l)d\mu(x^l)\nonumber\\
&\stackrel{(a)}{\geq } \int\liminf_{b\to\infty}{1\over b}H([\sum_{i=1}^la_ix_{l-i}+U]_b|X^l=x^l)d\mu(x^l)\nonumber\\
&\stackrel{(b)}{= }1,\label{eq:35}
\end{align}
where (a) follows from the Fatou's Lemma, and $(b)$ follows because $U$ has an absolutely continuous distribution. On the other hand,
\begin{align}
\liminf_{b\to\infty} {H([\sum_{i=1}^la_iX_{l-i}+U]_b|X^l)\over b}\leq \limsup_{b\to\infty} {H([\sum_{i=1}^la_iX_{l-i}+U]_b|X^l)\over b}\leq 1,
\end{align} 
where the last inequality follows because ${H([\sum_{i=1}^la_iX_{l-i}+U]_b|X^l)\over b}\leq 1$, for all $b$. Therefore, combining this with \eqref{eq:35} yields the desired result. That is, $\lim_{b\to\infty} {H([\sum_{i=1}^la_iX_{l-i}+U]_b|X^l)\over b}=1$, and therefore, $\limsup_{b}{H([X_{l+1}]_b|X^l)\over b}\geq p$ and $\liminf_{b}{H([X_{l+1}]_b|X^l)\over b}\geq p$. These lower bounds combined with the upper bounds derived earlier prove that $\bar{d}_o(X)=\underline{d}_o(X)=p$.


\subsection{Proof of Theorem \ref{thm:moving-ave}}
Define process $Z$ as an indicator of the locations where $Y$ is non-zero. That is, $Z_i=\ind_{Y_i\neq 0}$. Then,
\begin{align}
\bar{d}_k(X)&=\limsup_{b}{H([X_k]_b|[X^{k-1}]_b)\over b}\nonumber\\
&\leq \limsup_{b}{H([X_k]_b,Z_{k-1}|[X^{k-1}]_b)\over b}.
\end{align}
But since process $Y$ is an i.i.d. process and $X^{k-1}$ only depends on $Y_{-l+1}^{k-2}$, $Z_{k-1}$ is independent of $[X^{k-1}]_b$ and  therefore
\begin{align} 
{H([X_k]_b,Z_{k-1}|[X^{k-1}]_b)\over b}&={h(p)\over b}+p{H([X_k]_b|Z_{k-1}=1,[X^{k-1}]_b)\over b}\nonumber\\
&\;\;\;+(1-p){H([X_k]_b|Z_{k-1}=0,[X^{k-1}]_b)\over b}.
\end{align}
But, 
\begin{align}
\limsup_{b} {H([X_k]_b|Z_{k-1}=1,[X^{k-1}]_b)\over b}\leq \limsup_{b} {H([X_k]_b)\over b}\leq 1.
\end{align}
Therefore, 
\begin{align}
\bar{d}_k(X)&\leq p+(1-p)\limsup_{b} {H([X_k]_b|Z_{k-1}=0,[X^{k-1}]_b)\over b}. \label{eq:bd-dk-main}
\end{align}
We next prove that $\lim_{k\to\infty}\limsup_{b} {H([X_k]_b|Z_{k-1}=0,[X^{k-1}]_b)\over b}=0.$ 
Note that
\begin{align}
H([X_k]_b|Z_{k-1}=0,[X^{k-1}]_b)&=H([X_k]_b|Z_{k-1}=0,[X^{k-1}]_b,Y_{-l+1}^{-1})+I(Y_{-l+1}^{-1};[X_k]_b|Z_{k-1}=0,[X^{k-1}]_b).\label{eq:main-bd-Xk}
\end{align}
To prove the desired result, we first show that 
\begin{align}
\limsup_b{1\over b}H([X_k]_b|Z_{k-1}=0,[X^{k-1}]_b,Y_{-l+1}^{-1})=0.\label{eq:bd-first-X-k1}
\end{align}
Let $\Yh_i$, $i=0,1,\ldots,k-2$, denote an estimated value of $Y_i$,  as a function of $([X^{k-1}]_b,Y_{-l+1}^{-1})$, defined as follows. Since $X_i={1\over l}{\sum_{j=1}^l}Y_{i-j}$, we have $Y_i=lX_{i+1}-\sum_{j=1}^{l-1}Y_{i-j}$. From this equality, we define 
\[
\Yh_0=l[X_1]_b-\sum_{j=-l+1}^{-1}Y_{i-j},
\] 
\[
\Yh_i=l[X_{i+1}]_b-\sum_{j=i-l}^{i-2}Y_j-\sum_{i=0}^{i-1}\Yh_j,
\] 
for $i=1,\ldots,l-1$, and  
\[
\Yh_i=l[X_{i+1}]_b-\sum_{j=1}^{l-l}\Yh_{i-j},
\] 
for $i\geq l.$ Define the estimation error process as  $E_i=Y_i-\Yh_i$. For $i=0$,
\begin{align}
E_0&=l[X_1]_b-\sum_{j=-l+1}^{-1}Y_{i-j} -( lX_1-\sum_{j=-l+1}^{-1}Y_{i-j} )\nonumber\\
&= l([X_1]_b-X_1). 
\end{align}
Therefore, $|E_0|\leq l2^{-b}$. For $i=1,\ldots,l-1$,
\[
|E_i|\leq l2^{-b}+\sum_{i=0}^{i-1}|E_i|,
\]
and for $i\geq l$,
\begin{align}
|E_i|= |Y_i-\Yh_i|&\leq l|[X_{i+1}]_b-X_{i+1}|+\sum_{j=1}^{l-l} |E_{i-j}|\nonumber\\
&\leq l2^{-b}+\sum_{j=1}^{l-l} |E_{i-j}|.
\end{align}
We prove by induction that  $|E_i|\leq 2^il 2^{-b}$, for all $i$. Assume that we know that for $j=0,1,\ldots,i$, $|E_j|\leq 2^jl 2^{-b}$. Let $E_{j}=0$, for $j<0$. Then, for $i+1$,
\[
|E_{i+1}|\leq  l2^{-b} + \sum_{j=1}^{l}|E_{i+1-j}| \leq l2^{-b}+ l2^{-b} \sum_{j=1}^{l} 2^{i+1-j}\leq l2^{-b} \sum_{j=1}^{i}2^{j}= l2^{-b}(2^{i+1}-1)\leq 2^{i+1}l2^{-b}.
\]
Given the estimates $\{\Yh_i\}_i$, and since $X_k={1\over l}{\sum_{j=1}^l}Y_{k-j}$, define an estimate of $X_k$ as a function of $([X^{k-1}]_b,Y_{-l+1}^{-1})$ as follows
\[
\Xh_{k}=\sum_{j=1}^l\Yh_{k-j}.
\] Then, by the triangle inequality, 
\[
|X_k-\Xh_k|\leq \sum_{j=1}^l|\Yh_{k-j}-Y_{k-j}|=\sum_{j=1}^l |E_{k-j}|\leq {1\over l} \sum_{j=1}^l 2^{k-j}l 2^{-b}\leq 2^{k}l2^{-b}.
\]
This proves that given $([X^{k-1}]_b,Y_{-l+1}^{-1})$ there exists an estimate of $X_k$, whose distance to $X_k$ can be bounded by $2^{k}l2^{-b}$. Therefore, since
$H([X_k]_b|Z_{k-1}=0,[X^{k-1}]_b,Y_{-l+1}^{-1})=H([\Xh_k+X_k-\Xh_k]_b|Z_{k-1}=0,[X^{k-1}]_b,Y_{-l+1}^{-1})$, the remaining ambiguity in $[X_k]_b$ given $([X^{k-1}]_b,Y_{-l+1}^{-1})$ can be bounded by 
\[
\log{2^{k}l2^{-b}\over 2^{-b}}=k+\log l,
\]
which proves \eqref{eq:bd-first-X-k1}.

We next prove that 
\[
\limsup_{b}{1\over b}{I(Y_{-l+1}^{-1};[X_k]_b|Z_{k-1}=0,[X^{k-1}]_b)\over b}\leq (1-p)^{k-l}.
\]
Define for any $k>l$, define random variable $J_k$ as an indicator function, which is equal to one if there exists $j$ such that $l+1\leq j\leq k$, and $Y_{j-l+1}^j=0_k$. To bound $I(Y_{-l+1}^{-1};[X_k]_b|Z_{k-1}=0,[X^{k-1}]_b)$, note that 
\begin{align}
I(Y_{-l+1}^{-1};[X_k]_b|Z_{k-1}=0,[X^{k-1}]_b)&\leq I(Y_{-l+1}^{-1},J_k;[X_k]_b|Z_{k-1}=0,[X^{k-1}]_b)\nonumber\\
&\leq 1+I(Y_{-l+1}^{-1};[X_k]_b|J_k,Z_{k-1}=0,[X^{k-1}]_b)\nonumber\\
&\leq 1+I(Y_{-l+1}^{-1};[X_k]_b|J_k=0,Z_{k-1}=0,[X^{k-1}]_b)\P(J_k=0)\nonumber\\
&\;\;+I(Y_{-l+1}^{-1};[X_k]_b|J_k=1,Z_{k-1}=0,[X^{k-1}]_b)\P(J_k=1).
\end{align}
But conditioned on $J_k=1$ and $Z_{k-1}=0$,  $Y_{-l+1}^{-1}\to [X^{k-1}]_b \to [X_k]_b$, and therefore $I(Y_{-l+1}^{-1};[X_k]_b|J_k=1,Z_{k-1}=0,[X^{k-1}]_b)=0$. On the other hand, 
\[
I(Y_{-l+1}^{-1};[X_k]_b|J_k=0,Z_{k-1}=0,[X^{k-1}]_b)\leq b.
\]
Moreover, dividing $Y^k$ into non-overlapping blocks of size $l$, it follows that if there is no all-zero block of size $l$, then none of these non-overlapping blocks is all-zero. Now since these blocks are non-overlapping and process $Y$ is i.i.d., it shows that
\[
\P(J_k=0)\leq ((1-p)^l)^{\lfloor {k\over l} \rfloor}\leq (1-p)^{k-l}.
\]
Therefore, 
\[
\limsup_b {I(Y_{-l+1}^{-1};[X_k]_b|Z_{k-1}=0,[X^{k-1}]_b)\over b}\leq (1-p)^{k-l}.
\]
Combing this result with \eqref{eq:bd-dk-main}, \eqref{eq:main-bd-Xk} and \eqref{eq:bd-first-X-k1} proves that 
\[
\bar{d}_k(X)\leq p+(1-p)^{k-l+1},
\]
which as $k$ grows to infinity proves that $\bar{d}_o(X)\leq p$.

\subsection{Proof of Theorem \ref{thm:3}}

We show that for any $\e>0$,
\[
\P({1\over \sqrt{n}}\|X_o^n-\Xh_o^n\|_2>\e)\to 0,
\]
as $n\to \infty$.
Let $X_o^n=[X_o^n]_b+q_o^n$ and $\Xh_o^n=[\Xh_o^n]_b+\hat{q}_o^n$. By assumption, $AX_o^n=A\Xh_o^n$, and therefore, $A([X_o^n]_b-[\Xh_o^n]_b)=A(q_o^n-\qh_o^n)$. Note that
\begin{align}
\|A([X_o^n]_b-[\Xh_o^n]_b)\|_2=\|A(q_o^n-\qh_o^n)\|_2\leq \sigma_{\max}(A)\|q_o^n-\qh_o^n\|_2.\label{eq:basic-ineq-thm4}
\end{align}
Define the event
\[
\Ec_1\triangleq \{\sigma_{\max}(A)\leq \sqrt{n}+2\sqrt{m}\}.
\]
 From \cite{CaTa05}, $\P(\Ec_1^c)\leq \ex^{-m/2}$. But,
 \[
 \|q_o^n-\qh_o^n\|_2\leq \|q_o^n\|_2+\|\qh_o^n\|_2\leq \sqrt{n}2^{-b+1}.
 \]
 Hence, conditioned on $\Ec_1$,
\begin{align*}
\|A(q_o^n-\qh_o^n)\|_2&\;\leq\; \sigma_{\max}(A)\|q_o^n-\qh_o^n\|_2\\
&\;\leq\; n\Big(1+2\sqrt{m\over n}\Big)2^{-b+1}.
\end{align*}

As the next step, we derive a lower bound on $\|A([X_o^n]_b-[\Xh_o^n]_b)\|_2$, which holds with high probability. For a fixed vector $u^n$ and for any $\tau\in(0,1)$, by Lemma \ref{lemma:chi},
\begin{align}
\P(\|Au^n\|_2\leq \sqrt{m(1-\tau)}\|u^n\|_2)\leq \ex^{{m\over 2}(\tau+\ln(1-\tau))}.\label{eq:lower-bd-size-Au}
\end{align}
However, $[X_o^n]_b-[\Xh_o^n]_b$ is not a fixed vector. But as we will show next, with high probability, we can upper bound the number of such vectors possible.

Since $\Xh_o^n$ is the solution of \eqref{eq:MEP-ideal}, we have
\begin{align}
\hat{H}_k([\Xh_o^n]_b)\leq \hat{H}_k([X_o^n]_b).\label{eq:Hk-compare}
\end{align}
On the other hand as proved in Appendix \ref{app:D},
\begin{align}
{1\over n}\ell_{\rm LZ}([\Xh_o^n]_b)\leq \hat{H}_k([\Xh_o^n]_b)+{b(kb+b+3)\over (1-\e_n)\log n-b}+\gamma_n,\label{eq:LZ-ub-Hk-Xhat}
\end{align}
where $\gamma_n=o(1)$, and does not depend on  $[\Xh_o^n]_b$ or $b$, and
\[
\e_n={\log ((2^b-1)(\log n)/b+2^b-2)+2b\over \log n}.
\]
Combining \eqref{eq:Hk-compare} and \eqref{eq:LZ-ub-Hk-Xhat} and dividing both sides by $b=b_n$ yields
\begin{align}
{1\over n b_n}\ell_{\rm LZ}([\Xh_o^n]_{b_n})\leq {\hat{H}_k([X_o^n]_{b_n})\over b_n}+{kb_n+b_n+3\over (1-\e_n)\log n-b_n}+{\gamma_n\over b_n}.\label{eq:-bd-LZ-H-hat}
\end{align}

As the next step, we find an upper bound on $ \hat{H}_k([X_o^n]_{b_n})/b_n$ that holds with high probability.  The upper information dimension of process $X$, $\bar{d}_o(X)$, is defined as $\lim_{k\to\infty}\bar{d}_k(X)$. By Lemma \ref{lemma:1}, $\bar{d}_k(X)$ is a non-increasing function of $k$. Therefore,  given $\d_1>0$, there exists $k_{\d_1}>0$, such that for any $k>k_{\d_1}$,
\begin{align}
\bar{d}_o(X)\;\leq\; \bar{d}_k(X)\;\leq\; \bar{d}_o(X)+\d_1.\label{eq:bound-do-dk}
\end{align}
By the definition of $d_{k_{\d_1}}(X)$, given $\d_2>0$, there exists $b_{\d_2}$, such that for $b\geq b_{\d_2}$,
\begin{align}
{H([X_{k_{\d_1}+1}]_{b}|[X^{k_{\d_1}}]_{b})\over b} \leq d_{k_{\d_1}}(X)+\d_2.\label{eq:bd-dk-b-delta}
\end{align}
As a reminder, by the definition, $\hat{H}_k(x^n)=H(U_{k+1}|U^{k})$, where $U^{k+1}\sim {p}_{k+1}(\cdot|x^n)$. Therefore, $\hat{H}_{k}([X_o^n]_{b_n})/b_n$ is a decreasing function of $k$. Therefore, since $k=k_n$ by construction is a diverging sequence,  for $n$ large enough, $k_n>k_{\d_1}$, and
\[
{\hat{H}_{k_{n}}([X_o^n]_{b_n})\over b_n}\leq {\hat{H}_{k_{\d_1}}([X_o^n]_{b_n})\over b_n}.
\]
We now  prove that, given our choice of parameters, for large values of $n$, $\hat{H}_{k_{\d_1}}([X_o^n]_{b_n})/b_n=H(U_{k_{\d_1}+1}|U^{k_{\d_1}})/{b_n}$, where $U^{k_{\d_1}+1}\sim p_{k_{\d_1}+1}(\cdot|[X_o^n]_{b_n})$, is close to
\[
{H([X_{k_{\d_1}+1}]_{b_n}|[X^{k_{\d_1}}]_{b_n})\over b_n},
\]
with high probability.

Since $X$ is $\Psi^*$-mixing process, by Theorem  \ref{thm:exp-rate-markov}, given $\e_1>0$, there exists $g\in\mathds{N}$, only depending on $\e_1$ and the distribution of the source process, such that for any  $n>6(k+g)/\e_1+k$,
\[
\P( \|p_k(\cdot|[X^n]_b)-\mu_k\|_1\geq \e_1)\leq2^{c\e_1^2/8} (k+g) n^{2^{kb}}2^{-nc\e_1^2\over 8(k+g)},
\]
where $c=1/(2\ln 2)$. (The  empirical distribution function $p_{k_{\d_1+1}}(\cdot|[X^n]_{b_n})$ is defined in Section \ref{sec:cond-ent}.) Let
\[
\Ec_2\triangleq \{\|p_{k_{\d_1+1}}(\cdot|[X^n]_{b_n})-\mu_{k_{\d_1+1}}\|_1\leq \e_1/(k_{\d_1}+1)\}.
\]  Letting $k=k_{\d_1}+1$, where $k_{\d_1}>l$, and $b=b_n$, for $n$ large enough,
\begin{align}
\P(\Ec_2^c)\leq2^{c\e_1^2\over 8(k_{\d_1}+1)} (k_{\d_1}+g+1) n^{2^{b_n(k_{\d_1}+1)}}2^{-nc\e_1^2\over 8(k_{\d_1}+g+1)(k_{\d_1}+1)^2}.\label{eq:bd-Ec-2}
\end{align}
Let  $U^{k_{\d_1}+1}\sim p_{k_{\d_1}+1}(\cdot|[x^n]_{b_n})$. Then, conditioned on $\Ec_2$,
\begin{align}
|\hat{H}_{k_{\d_1}}([x^n]_{b_n})-H([X_{k_{\d_1}+1}]_{b_n}|[X^{k_{\d_1}}]_{b_n})|&=|H(U_{k_{\d_1}}|U^{k_{\d_1}+1})-H([X_{k_{\d_1}+1}]_{b_n}|[X^{k_{\d_1}}]_{b_n})\nonumber\\
&=|H(U^{k_{\d_1}+1})-H(U^{k_{\d_1}})-H([X^{k_{\d_1}+1}]_{b_n})+H([X^{k_{\d_1}}]_{b_n})|\nonumber\\
&\leq |H(U^{k_{\d_1}+1})-H([X^{k_{\d_1}+1}]_{b_n})|+|H(U^{k_{\d_1}})-H([X^{k_{\d_1}}]_{b_n})|\nonumber\\
&\stackrel{(a)}{\leq} -{2\e_1\over k_{\d_1}+1}\log ({\e_1\over k_{\d_1}+1}) +\e_1 b_n,\label{eq:H-hat-k-delta1}
\end{align}
where (a) follows from  Lemma \ref{lemma:bd-TV-dist}. Dividing both sides of \eqref{eq:H-hat-k-delta1} by $b_n$ yields
\begin{align}
\Big|{\hat{H}_{k_{\d_1}}([x^n]_{b_n})\over b_n}-{H([X_{k_{\d_1}+1}]_{b_n}|[X^{k_{\d_1}}]_{b_n})\over b_n}\Big|&\leq -{2\e_1\over (k_{\d_1}+1)b_n}\log ({\e_1\over k_{\d_1}+1}) +\e_1.\label{eq:bd-hat-H-over-bn}
\end{align}
On the other hand, $b_n$ is a diverging sequence of $n$. Therefore, for $n$ large enough, $b_n\geq b_{\d_2}$, and as a result, combining \eqref{eq:bound-do-dk}, \eqref{eq:bd-dk-b-delta} and \eqref{eq:bd-hat-H-over-bn} yields that, for $n$ large enough, conditioned on $\Ec_2$,
\begin{align}
{\hat{H}_{k_n}([X_o^n]_{b_n})\over b_n}\leq \bar{d}_o(X)+\d_3,\label{eq:bd-hat-H-vs-do}
\end{align}
where $\d_3\triangleq \d_1+\d_2- {2\e_1\over (k_{\d_1}+1)b_n}\log {\e_1\over k_{\d_1}+1} +\e_1$  can be made arbitrarily small by choosing $\e_1$, $\d_1$ and $\d_2$ small enough. Furthermore, given our choice of parameters $b_n$ and $k_n$,  from \eqref{eq:-bd-LZ-H-hat} and \eqref{eq:bd-hat-H-vs-do}, conditioned on $\Ec_2$,
\begin{align}
{1\over nb_n}\ell_{\rm LZ}([\Xh_o^n]_{b_n})\leq  \bar{d}_o(X)+\d_4,\label{eq:upper-bd-LZ-1}
\end{align}
and
\begin{align}
{1\over nb_n}\ell_{\rm LZ}([X_o^n]_{b_n})\leq  \bar{d}_o(X)+\d_4,\label{eq:upper-bd-LZ-1}
\end{align}
where $\d_4\triangleq {(k_nb_n+b_n+3)/ ((1-\e_n)\log n-b_n)}+{\gamma_n/ b_n}+\d_3$ can be made arbitrarily small.

Let $\Cc_n\triangleq \{[x^n]_{b_n}: \; \ell_{\rm LZ}([x^n]_{b_n})\leq nb_n(\bar{d}_o(X)+\d_4)\}$. Since the Lempel-Ziv code is a uniquely decodable code, each binary string corresponding to an LZ-coded sequence corresponds to a unique uncoded sequence. Hence, the number of sequences in $\Cc_n$, \ie the number of quantized sequences satisfying the upper bound of \eqref{eq:upper-bd-LZ-1}, can be  bounded as
\begin{align}
|\Cc_n|&\leq \sum_{i=1}^{ nb_n(\bar{d}_o(X)+\d_4)}2^i\nonumber\\
&\leq 2^{nb_n(\bar{d}_o(X)+\d_4)+1}.\label{eq:bd-size-Cc}
\end{align}
Define the event $\Ec_3$ as follows:
\[
\Ec_3\triangleq\{\|A([X_o^n]_{b_n}-[x^n]_{b_n})\|_2\geq \sqrt{m(1-\tau)}\|[X_o^n]_{b_n}-[x^n]_{b_n}\|_2; \forall \; x^n\in[0,1]^n, [x^n]_{b_n}\in\Cc_n\}.
\]
Assume that we fix $X_o^n$ and only consider the randomness  in drawing the matrix $A$. Then, applying the union bound to  \eqref{eq:lower-bd-size-Au}  and noting the upper bound on the size of $\Cc_n$, derived in \eqref{eq:bd-size-Cc}, we get
\begin{align}
\P_A(\Ec_3^{c})\leq 2^{nb_n(\bar{d}_o(X)+\d_4)+2}\ex^{{m\over 2}(\tau+\ln(1-\tau))}.\label{eq:bd-Ec-Xo-n}
\end{align}
Since $X_o^n$ is a random vector, $\P_A(\Ec_3^{c})$ is a random variable depending on $X_o^n$.  Taking the expectation of  both sides of \eqref{eq:bd-Ec-Xo-n}, it follows that
\begin{align}
\E_{X_o^n}[\P_A(\Ec_3^{c})]\leq 2^{nb_n(\bar{d}_o(X)+\d_4)+2}\ex^{{m\over 2}(\tau+\ln(1-\tau))}.\label{eq:Ec-Xo-n-bd}
\end{align}
We next prove with the right choice of parameters,  the right hand side of  \eqref{eq:Ec-Xo-n-bd} goes to zero. Let
\[
\tau={1-{1\over (\log n)^{2/(1+\upsilon)} }},
\]
where $\upsilon\in(0,1)$. Then, since by assumption $b=b_n=\lceil \log\log n\rceil $, and $m=m_n\geq (1+\d)\bar{d}_o(X)n$, it follows that
\begin{align}
\E_{X_o^n}[\P_A(\Ec_3^{c})]&\leq 2^{nb_n(\bar{d}_o(X)+\d_4)+2}\ex^{{m\over 2}(1-{2\over 1+\upsilon} \ln \log n})\nonumber\\
&\leq 2^{n(\log \log n+1)(\bar{d}_o(X)+\d_4)+2}2^{0.5(1+\d)\bar{d}_o(X)n(\log \ex -{2\over 1+\upsilon} \log \log n)}\nonumber\\
&=2^{n(\log \log n) (1+\d)\bar{d}_o(X)({1+\d_5\over 1+\d}-{1\over 1+\upsilon }+\rho_n) }, \label{eq:bd-Ec-3-cond}
\end{align}
where $\d_5\triangleq \d_4/\bar{d}_o(X)$ and $\rho_n=o(1)$. Note that $\d_4$ and as a result $\d_5$ can be made arbitrarily small for $n$ large enough. Given $\d>0$, we choose $\delta_4$ such that $\delta_5<\delta$ and set $\upsilon=0.5(\d-\d_5)/(1+\d_5)$. Then, since for this choice of parameters for $n$ large enough,
\[
(1+\d)\Big({1+\d_5\over 1+\d}-{1\over 1+\upsilon }+\rho_n\Big)\leq-\Big({\delta-\delta_5\over 4}\Big),
\]
from \eqref{eq:bd-Ec-3-cond}, we have
\begin{align}
\E_{X_o^n}[\P_A(\Ec_3^{c})]&\leq 2^{-n(\log \log n)\bar{d}_o(X)(\delta-\d_5)/4  }. \label{eq:bd-Ec-3-cond-final}
\end{align}

On the other hand,
\begin{align}
\E_{X_o^n}[\P_A(\Ec_3^{c})]&=\E_{X_o^n}[\E_{A}[\ind_{\Ec_3^{c}}]]\nonumber\\
&=\E_{A}[\E_{X_o^n}[\ind_{\Ec_3^{c}}]]\nonumber\\
&=\E_{A}[\P_{X_o^n}(\Ec_3^{c})],\label{eq:fubini}
\end{align}
where the first step follows from  Fubini's Theorem. We next prove that $\P_{X_o^n}(\Ec_3^{c})$ converges to zero, almost surely. To prove this result, we employ the Borel Cantelli Lemma. By the Markov inequality, for any $\e>0$, $\P_A(\P_{X_o^n}(\Ec_3^{c})>\e)\leq \e^{-1}2^{-n(\log \log n)\bar{d}_o(X)(\delta-\d_5)/4  }$. Since  $\sum_{n=1}^{\infty}\P_A(\P_{X_o^n}(\Ec_3^{c})>\e)<\infty$, by the Borel Cantelli Lemma, as $n$ grows to infinity, $\P_{X_o^n}(\Ec_3^{c})$ converges to zero,  almost surely.

By the union bound, $\P((\Ec_1\cap\Ec_2\cap\Ec_3)^c)\leq \P(\Ec_1^c)+\P(\Ec_2^c)+\P(\Ec_3^c)$. Clearly $\P(\Ec_1^c)\to 0$, as $n\to\infty$.  Also, for our choice of parameter $b=b_n$, from \eqref{eq:bd-Ec-2}, as $n\to\infty$, $\P(\Ec_2^c)\to 0$ as well. Finally,  as we just proved, $\P_{X_o^n}(\Ec_3^c)$ converges to zero, almost surely.
But conditioned on $\Ec_1\cap\Ec_2\cap\Ec_3$, since $\Xh_o^n\in \Cc_n$, from \eqref{eq:basic-ineq-thm4},
\[
\sqrt{m(1-\tau)}\|X_o^n-\Xh_o^n\|_2\leq  n\Big(1+2\sqrt{m\over n}\Big)2^{-b_n+1},
\]
or
\begin{align}
{1\over \sqrt{n}}\|X_o^n-\Xh_o^n\|_2\leq \sqrt{ n\over m (1-\tau)}\Big(1+2\sqrt{m\over n}\Big)2^{-b_n+1}.\label{eq:bd-error-1}
\end{align}
Since $m/n\leq 1$, we have
\begin{align}
{1\over \sqrt{n}}\|X_o^n-\Xh_o^n\|_2&\leq  {3(\log n)^{1/(1+\upsilon)}\over \sqrt{2(1+\d)\bar{d}_o(X)}} 2^{-\log\log n+1}\nonumber\\
&\leq  {3(\log n)^{1/(1+\upsilon)}\over \sqrt{2(1+\d)\bar{d}_o(X)}} 2^{-\log\log n+1}\nonumber\\
&\leq  {6\over \sqrt{2(1+\d)\bar{d}_o(X)}} 2^{-\log\log n(1-1/(1+\upsilon))},
\end{align}
which goes to zero as $n$ grows to infinity. This concludes  the proof.

\subsection{Proof of Theorem \ref{thm:4}}

We need to prove that for any $\e>0$,
\[
\P({1\over \sqrt{n}}\|X_o^n-\Xh_o^n\|_2>\e)\to 0,
\]
as $n\to \infty$.
Throughout the proof $\bar{d}_o$ refers to $\bar{d}_o(X)$. As before, let $X_o^n=[X_o^n]_b+q_o^n$. As we showed earlier, $\|q_o^n\|_2\leq \sqrt{n}2^{-b}$. Since $\Xh_o^n=\argmin_{u^n\in\Xc_b^n} (\hat{H}_{k}(u^n)+{\lambda\over n^2} \|Au^n-Y_o^m\|_2)$, we  have
\begin{align}
\hat{H}_{k}(\Xh_o^n)+{\lambda\over n^2} \|A\Xh_o^n-Y_o^m\|_2^2&\leq \hat{H}_{k}([X_o^n]_{b})+{\lambda\over n^2}  \|Aq_o^n\|_2^2,\nonumber\\
&\leq \hat{H}_{k}([X_o^n]_{b})+{\lambda(\sigma_{\max}(A))^22^{-2b}\over n} . \label{eq:min-lagrangian-cost-1}
\end{align}
Define event $\Ec_1$ as $\Ec_1\triangleq \{\sigma_{\max}(A)\leq \sqrt{n}+2\sqrt{m}\}$, where  from \cite{CaTa05}, $\P(\Ec_1^c)\leq \ex^{-m/2}$. Also given $\e>0$, define event $\Ec_2$ as
\[
\Ec_2\triangleq \{{1\over b}\hat{H}_k([X_o^n]_b)\leq \bar{d}_o+\e\}.
\]
In the proof of Theorem \ref{thm:3}, we showed that, for any $\e>0$, given our choice of parameters, $\P(\Ec_2^c)$ converges to zero as $n$ grows to infinity. Conditioned on $\Ec_1\cap\Ec_2$, from \eqref{eq:min-lagrangian-cost-1}, we derive
\begin{align}
{1\over b}\hat{H}_{k}(\Xh_o^n)+{\lambda\over bn^2} \|A\Xh_o^n-Y_o^m\|_2^2 &\leq \bar{d}_o+\e+{\lambda2^{-2b}\over b}(1+2\sqrt{m/n})^2 \nonumber\\
&\leq \bar{d}_o+\e+{9\lambda2^{-2b}\over b}, \label{eq:min-lagrangian-cost-2}
\end{align}
where the last line holds since for $m\leq n$, $1+2\sqrt{m/n}\leq 3.$ On the other hand, for  $b=b_n=\lceil r\log\log n\rceil $ and $\l=\l_n=(\log n)^{2r}$, we have
\[
{9\lambda2^{-2b}\over b}\leq {9(\log n)^{2r}\over r(\log n)^{2r}\log\log n }={9\over r \log \log n},
\]
which goes to zero as $n$ grows to infinity. For $n$ large enough, $9/(r\log\log n)\leq \e$, and hence from \eqref{eq:min-lagrangian-cost-2},
\[
{1\over b}\hat{H}_{k}(\Xh_o^n)+{\lambda\over bn^2} \|A\Xh_o^n-Y_o^m\|_2^2 \leq \bar{d}_o+2\e,
\]
which implies that
\begin{align}
{1\over b}\hat{H}_{k}(\Xh_o^n) &\leq \bar{d}_o+2\e,\label{eq:upper-bd-Hk}
\end{align}
and since $\bar{d}_o\leq 1$,
\begin{align}
\sqrt{\lambda\over bn^2}\; \|A\Xh_o^n-Y_o^m\|_2 &\leq \sqrt{1+2\e}.\label{eq:upper-bd-dist}
\end{align}

For $x^n\in[0,1]^n$, from \eqref{eq:LZ-ub-Hk} in Appendix \ref{app:D},
\[
{1\over n}\ell_{\rm LZ}([x^n]_b)\leq \hat{H}_k([x^n]_b)+{b(kb+b+3)\over (1-\e_n)\log n-b}+\gamma_n,
\]
where  $\e_n=o(1)$ and $\g_n=o(1)$ are both independent of $x^n$.  Therefore there exists $n_{\e}$, such that for $n>n_{\e}$,  \[{1\over nb}\ell_{\rm LZ}([x^n]_b)\leq {1\over b}\hat{H}_k([x^n]_b)+\e.\] On the other hand, conditioned on $\Ec_1\cap\Ec_2$,  ${1\over b}\hat{H}_k(X_o^n) \leq \bar{d}_o+\e$, and ${1\over b}\hat{H}_k(\Xh_o^n) \leq \bar{d}_o+2\e$. Therefore, conditioned on $\Ec_1\cap\Ec_2$, $[X_o^n]_b, \Xh_o^n\in\Cc_n$, where
\[
\Cc_n\triangleq \{[x^n]_{b_n}: \; {1\over nb} \ell_{\rm LZ}([x^n]_{b_n})\leq \bar{d}_o+3\e\}.
\]
Define the event $\Ec_3$ as
\[
\Ec_3\triangleq \{\|A(u^n-[X_o^n]_b)\|_2\geq \|u^n-[X_o^n]_b\|_2\sqrt{(1-\tau)m}: \forall u^n\in\Cc_n\},
\]
where $\tau>0$. As we argued in the proof of Theorem \ref{thm:3}, for a fixed input vector $X_o^n$, by the union bound, we have
\[
\P_A( \Ec_3^c)\leq 2^{(\bar{d}_o+3\e)bn}\ex^{{m\over 2}(\tau+\ln(1-\tau))}.
\]
Let $\tau=1-(\log n)^{-{2r\over 1+f}}$, where $f>0$. Then, since $b=b_n\leq r\log\log n+1$, and $m=m_n>2(1+\d)\bar{d}_on$,
\begin{align}
\P_A( \Ec_3^c)&\leq 2^{(\bar{d}_o+3\e)(r\log\log n+1)n}2^{{m\over 2}(\log \ex-{2r\over 1+f}\log \log n)}\nonumber\\
&= 2^{2r\log\log n\left(n(\bar{d}_o+3\e)-{m\over 2(1+f)}+\alpha_n\right)}\nonumber\\
&\leq 2^{r(\log\log n)n\left(\bar{d}_o+3\e-({1+\d \over 1+f})\bar{d}_o+{1\over n}\alpha_n\right)},
\end{align}
where $\alpha_n=o(1)$. For $0<f<\delta$ and  $\e<(\d-f)\bar{d}_o/6(1+f)$, for $n$ large enough, $\bar{d}_o+3\e-({1+\d \over 1+f})\bar{d}_o+{1\over n}\alpha_n< -(\d-f)/2(1+f)$.  Let $\nu\triangleq (\d-f)/2(1+f)$. Then,  for $n$ large enough, $\P_A( \Ec_3^c)< 2^{-r\nu(\log\log n)n}$. Now applying Fubini's Theorem and the Borel Cantelli Lemma, similar to the proof of Theorem \ref{thm:3}, it follows that as $n$ grows to infinity, $\P_{X_o^n}( \Ec_3^c)\to 0$, almost surely.

Conditioned on $\Ec_1\cap\Ec_2\cap\Ec_3$,  $\Xh_o^n\in\Cc_n$. Moreover, by the triangle inequality, $\|A\Xh_o^n-Y_o^m\|_2 \geq \|A(\Xh_o^n-[X_o^n]_b)\|_2-\|Aq_o^n\|_2$.  Therefore, conditioned on $\Ec_1\cap\Ec_2\cap\Ec_3$,  it follows from  \eqref{eq:upper-bd-dist} that
\begin{align}
\sqrt{\lambda(1-\tau)m\over bn^2}\; \|\Xh_o^n-[X_o^n]_b\|_2& \leq \sqrt{\lambda\over bn^2}\|Aq_o^n\|_2+ \sqrt{1+2\e}\nonumber\\
& \leq \sqrt{\lambda\over b}(1+2\sqrt{m/n})2^{-b}+ \sqrt{1+2\e}\nonumber\\
& \leq \sqrt{9\lambda\over b}2^{-b}+ \sqrt{1+2\e},
\end{align}
or
\[
{1\over \sqrt{n}}\|\Xh_o^n-[X_o^n]_b\|_2 \leq 2^{-b}\sqrt{9n\over (1-\tau)m}+ \sqrt{(1+2\e)bn\over (1-\tau)\lambda m}.
\]
Hence, for $\tau=1-(\log n)^{-{2r\over 1+f}}$, $\lambda=(\log n)^{2r}$, and $b=\lceil r\log\log n\rceil$,
\begin{align}
{1\over \sqrt{n}}\|\Xh_o^n-[X_o^n]_b\|_2 \leq {3+ \sqrt{(1+2\e) (r\log\log n+1)} \over\sqrt{\bar{d}_o(1+\d ) } (\log n)^{rf\over 1+f}},
\end{align}
which can be made arbitrarily small.

\subsection{Proof of Theorem \ref{thm:5}}

The proof is very similar to the proof of Theorem \ref{thm:4}. Throughout the proof, for  ease of notation, $\bar{d}_o(X)$ is  denoted by $\bar{d}_o$.  Let $X_o^n=[X_o^n]_b+q_o^n$, $\e>0$, $\tau>0$ and $\Cc_n\triangleq \{[x^n]_{b}: \; {1\over nb} \ell_{\rm LZ}([x^n]_{b})\leq \bar{d}_o+3\e\}$. Define events $\Ec_1$, $\Ec_2$ and $\Ec_3$ as done in the proof of Theorem \ref{thm:4}, \ie  $\Ec_1\triangleq \{\sigma_{\max}(A)\leq \sqrt{n}+2\sqrt{m}\}$, and $\Ec_2\triangleq \{{1\over b}\hat{H}_k([X_o^n]_b)\leq \bar{d}_o+\e\}.$ and $\Ec_3\triangleq \{\|A(u^n-[X_o^n]_b)\|_2\geq \|u^n-[X_o^n]_b\|_2\sqrt{(1-\tau)m}: \forall u^n\in\Cc_n\}$. Also, define event $\Ec_4$ as
\[
\Ec_4\triangleq \{ \|Z^m\|_2\leq c_m\}.
\]

Since $\Xh_o^n$ is the minimizer of the cost function in \eqref{eq:MEP-lagrangian}, we have
\begin{align}
\hat{H}_{k}(\Xh_o^n)+{\lambda\over n^2} \|A\Xh_o^n-Y_o^m\|_2^2&\leq \hat{H}_{k}([X_o^n]_{b})+{\lambda\over n^2}  \|Aq_o^n+Z^m\|_2^2. \label{eq:min-lag-cost-noisy-1}
\end{align}
But $\|Aq_o^n+Z^m\|_2^2\leq \|Aq_o^n\|_2^2+\|Z^m\|_2^2+2\|Aq_o^n\|_2\|Z^m\|_2\leq (\sigma_{\max}(A))^2\|q_o^n\|_2^2+\|Z^m\|_2^2+2\sigma_{\max}(A)\|q_o^n\|_2\|Z^m\|_2$.   Since $\|q_o^n\|_2\leq \sqrt{n}2^{-b}$ and $m\leq n$,  conditioned on $\Ec_1\cap\Ec_2\cap\Ec_4$, we get
\begin{align}
\hat{H}_{k}(\Xh_o^n)+{\lambda\over n^2} \|A\Xh_o^n-Y_o^m\|_2^2&\leq \hat{H}_{k}([X_o^n]_{b})+{\lambda\over n^2} \Big((\sqrt{n}+2\sqrt{m})^2n2^{-2b}+c_m^2+2c_m(\sqrt{n}+2\sqrt{m})\sqrt{n}2^{-b}\Big)\nonumber\\
&= \hat{H}_{k}([X_o^n]_{b})+{\lambda} \Big((1+2\sqrt{m/ n})^22^{-2b}+\left({c_m\over n}\right)^2+2{c_m\over n}(1+2\sqrt{m/n})2^{-b}\Big)\nonumber\\
&\leq b(\bar{d}_o+\e)+{\lambda} \Big(9(2^{-2b})+\left({c_m\over n}\right)^2+{6c_m\over n}2^{-b}\Big).\label{eq:min-lag-cost-noisy-2}
\end{align}
Dividing both sides of \eqref{eq:min-lag-cost-noisy-2}, it follows that
\begin{align}
{1\over b}\hat{H}_{k}(\Xh_o^n)+{\lambda\over n^2b} \|A\Xh_o^n-Y_o^m\|_2^2&\leq \bar{d}_o+\e+{\l\over b} \Big(9(2^{-2b})+({c_m\over n})^2+{6c_m\over n}2^{-b}\Big).\label{eq:min-lag-cost-noisy-3}
\end{align}
By the theorem's assumption, $\l=\l_n=(\log n)^{2r}$ and $b=b_n=\lceil r\log\log n\rceil$. For this choice of the parameters,
\begin{align}
{\l\over b2^{2b}}\leq {(\log n)^{2r}\over  r(\log\log n) (\log n)^{2r}}={1\over r\log \log n},\label{eq:bd-1}
\end{align}
\begin{align}
{\l c_m^2\over bn^2}\leq {(\log n)^{2r} c_m^2\over r(\log \log n)n^2}= {1 \over r \log \log n}({c_m (\log n)^r \over n})^2,\label{eq:bd-2}
\end{align}
and finally
\begin{align}
{\l c_m\over bn2^b}\leq {(\log n)^{2r} c_m\over r(\log \log n)n(\log n)^{r} }= {1 \over r \log \log n}({c_m (\log n)^r \over n}).\label{eq:bd-3}
\end{align}
Since $c_m=O({m\over (\log m)^r})$, the right hand sides of \eqref{eq:bd-1}, \eqref{eq:bd-2} and \eqref{eq:bd-3} converge to zero as $n$ grows to zero. Therefore, for $n$ large enough,   ${\l\over b} (9(2^{-2b})+({c_m\over n})^2+{6c_m\over n}2^{-b})<\e$. The rest of the proof follows exactly as the proof of Theorem \ref{thm:4}.

\section{Conclusions}\label{sec:conclusion}

In this paper we have studied the problem of universal compressed sensing, \ie the problem of recovering ``structured'' signals from their under-determined set of random linear projections without having prior information about the structure of the signal. We have considered structured signals that are modeled by stationary processes.  We have generalized R\'enyi's information dimension and defined information dimension of a stationary process, as a measure of  complexity for such processes. We have also calculated the information dimension of some stationary processes, such as Markov processes used to model piecewise constant signals.

We then have introduced the MEP  optimization approach for universal compressed sensing. The optimization is based on Occam's Razor and among the signals satisfying the measurements constraints seeks the ``simplest'' signal.  The  complexity of a signal is measured in terms of the conditional empirical entropy of its quantized version, which normalized by the quantization level serves as an estimator of the information dimension of the source.  We have proved that, asymptotically,  for $X_o^n$ generated by a $\Psi^*$-mixing process $X$ with upper information dimension of $\bar{d}_o(X)$, MEP requires slightly more that $\bar{d}_o(X)n$ random linear measurements to recover the signal. We have also provided an implementable version of MEP, the Lagrangian-MEP algorithm,  which is identical to the heuristic algorithm proposed in \cite{BaronD:11} and \cite{BaronD:12}. The Lagrangian-MEP algorithm has the same asymptotic performance as the original MEP, and is also robust to the measurement noise.  For memoryless sources with a mixture of discrete-continuous distribution, this result shows that both MEP and Lagrangian-MEP achieve the fundamental limits proved in \cite{WuV:10}, and therefore there is no loss in the performance due to not knowing the source distribution.
\setcounter{equation}{0}
\renewcommand{\theequation}{\thesection.\arabic{equation}}

\appendices
%
%

\section{Connection between $\ell_{\rm LZ}$ and $\hat{H}_k$}\label{app:D}

In this appendix, we adapt the results of \cite{Ziv_inequality} to the case where the source is non-binary.  Since in this work in most cases we deal with  real-valued sources that are quantized at different number of quantization levels and the number of of quantization levels usually grows to infinity with blocklength, we need to derive all the constants carefully to make sure that the bounds are still valid, when the size of the alphabet depends on the blocklength.

Consider a finite-alphabet sequence $z^n\in\Zc^n$, where  $ |\Zc|=r$ and $r=2^{b}$, $b\in\mathds{N}.$\footnote{Restricting  the alphabet size to satisfy   this condition is to simplify the arguments, but the results can be generalized to any finite-alphabet source.} Let $N_{\rm LZ}(z^n)$ denote the number of phrases in $z^n$ after parsing it according the Lempel-Ziv algorithm \cite{LZ}.

For $k\in\mathds{N}$, let
\begin{align}
n_k=\sum_{j=1}^kjr^j={r^{k+1}(kr-(k+1))+r\over (r-1)^2}.\label{eq:n-k}
\end{align}
Let $N_{\rm LZ}(n)$ denote the maximum possible  number of phrases in a parsed sequence of length $n$. For $n=n_k$,
\begin{align}
N_{\rm LZ}(n_k)&\leq \sum_{j=1}^kr^j\nonumber\\
&={r(r^k-1)\over r-1}\nonumber\\
&={r^{k+1}(kr-(k+1))\over (r-1)(kr-(k+1))}-{r\over r-1}\nonumber\\
&\leq {(r-1)n_k\over kr-(k+1)}\nonumber\\
&\leq {n_k\over k-1}.\label{eq:upper-bd-nk}
\end{align}
Now given $n$,  assume that $n_k\leq n<n_{k+1}$, for some $k$.  It is straightforward to check that for $k\geq 2$, $n\geq n_k\geq r^k$. Therefore, if  $n>r(1+r)$,
\[
k\leq {\log n\over \log r}.
\]
On the other hand, $n<n_{k+1}$, where from \eqref{eq:n-k}
\begin{align*}
n_{k+1}&={r^{k+2}((k+1)r-(k+2))+r\over (r-1)^2}\\
&={r^{k+2}((r-1)\log_rn+r-2)+r\over (r-1)^2}.
\end{align*}
Therefore,
\begin{align*}
k+2\geq \log_r {n(r-1)^2-r\over (r-1)\log_r n+r-2},
\end{align*}
or
\begin{align}
k\geq (1-\e_n){\log n\over \log r},\label{eq:upper-bd-k}
\end{align}
where
\begin{align}
\e_n={\log (((r-1)\log_r n+r-2)r^2)\over \log n}.\label{eq:e-n}
\end{align}

To pack the maximum possible number of phrases in a sequence of length $n$, we need to first pack all possible phrases of length smaller than or equal to $k$, then use phrases of length $k+1$ to cover the rest. Therefore,
\begin{align}
N_{\rm LZ}(n)&\leq N_{\rm LZ}(n_k)+{n-n_k\over k+1}\nonumber\\
&\leq  {n_k\over k-1} + {n-n_k\over k+1}\nonumber\\
&\leq {n\over k-1}.\label{eq:ub-N-LZ}
\end{align}
Combining \eqref{eq:ub-N-LZ} with \eqref{eq:upper-bd-k}, and noting that $\log r=b$,  yields
\begin{align}
{N_{\rm LZ}(n)\over n}\leq {b\over (1-\e_n)\log n-b}.\label{eq:N-LZ-ub}
\end{align}

Taking into account the number of bits required for describing the blocklength $n$, the number of phrases $N_{\rm LZ}$,   the  pointers and the extra symbols of phrases, we derive
\begin{align}
{1\over n}\ell_{\rm LZ}(z^n)={1\over n}N_{\rm LZ}\log N_{\rm LZ}+ {b\over n}N_{\rm LZ}+\eta_n,\label{eq:LZ-norm}
\end{align}
where
\begin{align}
\eta_n={1\over n}(\log n+2\log \log n +\log N_{\rm LZ}+2\log\log N_{\rm LZ}+2), \label{eq:eta-n}
\end{align}

On the other hand,  straightforward extension of the analysis presented in \cite{Ziv_inequality} to the case of general non-binary alphabets yields
\begin{align}
{1\over n}N_{\rm LZ}\log N_{\rm LZ}\leq \hat{H}_k(z^n)+{N_{\rm LZ}\over n}((\mu+1)\log(\mu+1)-\mu\log \mu +k\log r),\label{eq:Ziv-ineq-ext}
\end{align}
where $\mu\triangleq N_{\rm LZ}/n$. But, $(\mu+1)\log(\mu+1)-\mu\log \mu=\log(\mu+1)+\mu\log(1+1/\mu)\leq \log(\mu+1)+1/\ln 2<\log \mu+2$. Also, it is easy to show that for any value of $r$ and $z^n$, $n\leq \sum_{i=1}^{N_{\rm LZ}}l$, or  $N_{\rm LZ}(z^n)\geq \sqrt{2n}-1$, or $n/N_{\rm LZ}(z^n)\leq \sqrt{n}$, for $n$ large enough. Therefore, since $\mu^{-1}\log \mu$ is an increasing function of $\mu$,
\[
{\log \mu \over \mu} \leq {\log n\over 2\sqrt{n}}.
\]
Hence, combining \eqref{eq:LZ-norm}, \eqref{eq:N-LZ-ub} and \eqref{eq:Ziv-ineq-ext}, we conclude that, for $n$ large enough,
\begin{align}
{1\over n}\ell_{\rm LZ}(z^n)\leq \hat{H}_k(z^n)+{b(kb+b+3)\over (1-\e_n)\log n-b}+\gamma_n,\label{eq:LZ-ub-Hk}
\end{align}
where $\gamma_n=\eta_n+{\log n\over 2\sqrt{n}}$, and $\e_n$ and $\eta_n$ are defined in \eqref{eq:e-n} and \eqref{eq:eta-n}, respecctively. Note that $\gamma_n=o(1)$  and does not depend on $b$ or $z^n$.

\bibliographystyle{unsrt}
\bibliography{myrefs}

\end{document}